\documentclass[10pt,journal,  oneside, twocolumn]{IEEEtran}

\usepackage{multirow}
\usepackage{latexsym}
\usepackage{graphicx}
\usepackage{float}
\usepackage{amsmath}
\usepackage{amsthm}
\usepackage{lipsum}
\usepackage{subfig}
\usepackage{graphicx}
\usepackage{authblk}
\usepackage{bm}
\usepackage{amsthm}
\usepackage[section]{placeins}
\usepackage{wasysym}

\usepackage{colortbl}

\usepackage{enumitem}

\newtheorem{theorem}{Theorem}

\newtheorem{lemma}{Lemma}

\newtheorem{definition}{Definition}

\makeatletter  
\newif\if@restonecol  
\makeatother

\usepackage[linesnumbered,ruled,vlined]{algorithm2e}%[ruled,vlined]{  
\usepackage{algpseudocode}  
\usepackage{amsmath}

\usepackage{amssymb}

\usepackage{mathrsfs}
\usepackage{subfig}
\usepackage{caption}
\begin{document}

\title{Stacked Intelligent Metasurfaces-Aided eVTOL Delay Sensitive Communications}

%Connected Autonomous Driving
% Resource Optimization for Driving Safety, Ride Comfort, and Road Traffic Throughput

\author{Liyuan Chen, Kai Xiong,~\IEEEmembership{Member,~IEEE}, Yujie Qin,~\IEEEmembership{Member,~IEEE}, Hanqing Yu,~\IEEEmembership{Member,~IEEE}, \\Supeng Leng,~\IEEEmembership{Member,~IEEE},
Chau Yuen,~\IEEEmembership{Fellow,~IEEE}

\thanks{%Copyright (c) 2015 IEEE. Personal use of this material is permitted. However, permission to use this material for any other purposes must be obtained from the IEEE by sending a request to pubs-permissions@ieee.org.

L. Chen, K. Xiong, Y. Qin, H. Yu, and S. Leng are with School of Information and Communication Engineering, University of Electronic Science and Technology of China, Chengdu, 611731, China; and, Shenzhen Institute for Advanced Study, University of Electronic Science and Technology of China, Shenzhen, 518110, China.
}

% \thanks{
% C. Huang is with College of Information Science and Electronic Engineering, Zhejiang University, Hangzhou 310027, China.
% }

\thanks{
C. Yuen is with School of Electrical and Electronics Engineering, Nanyang Technological University, 639798, Singapore. 
}

%\thanks{J. He  is with the School of Computer Science and Electronic Engineering, Colchester, CO4 3SQ, UK.}

%\thanks{C. Huang and C. Yuen are with Singapore University of Technology and Design (SUTD), Singapore.}

%\thanks{Y. L. Guan is with the School of Electrical and Electronic Engineering, Nanyang Technological University, Singapore.}

%\thanks{The financial support of National Natural Science Foundation of China (NSFC), Grant No.61374189, No. 61620106011 and No. 61601097, the joint fund of the Ministry of Education of China and China Mobile (MCM 20160304), the EU H2020 Project COSAFE (MSCA-RISE-2018-824019), the Fundamental Research Funds for the Central Universities, China (ZYGX2016Z011) and Zhongshan City Project (Grant No. 180809162197874).}

%\thanks{This work was partly supported by  .}

\thanks{The corresponding author is Kai Xiong, email: xiongkai@uestc.edu.cn}

}

% The paper headers
%\markboth{Journal of \LaTeX\ Class Files,~Vol.~6, No.~1, January~2007}%
%{Shell \MakeLowercase{\textit{et al.}}: Bare Demo of IEEEtran.cls for Journals}

% make the title area
\maketitle

%vehicular fog computing has emerged as a promising solution to accommodate the ever-increasing computational demands of Internet-of-Vehicles (IoV). However, 

\begin{abstract}
With rapid urbanization and increasing population density, urban traffic congestion has become a critical issue, and traditional ground transportation methods are no longer sufficient to address it effectively. To tackle this challenge, the concept of Advanced Air Mobility (AAM) has emerged, aiming to utilize low-altitude airspace to establish a three-dimensional transportation system. Among various components of the AAM system, electric vertical take-off and landing (eVTOL) aircraft plays a pivotal role due to their flexibility and efficiency. However, the immaturity of Ultra Reliable Low Latency Communication (URLLC) technologies poses significant challenges to safety-critical AAM operations. Specifically, existing Stacked Intelligent Metasurfaces (SIM)-based eVTOL systems lack rigorous mathematical frameworks to quantify probabilistic delay bounds under dynamic air traffic patterns, a prerequisite for collision avoidance and airspace management.
To bridge this gap, we employ network calculus tools to derive the probabilistic upper bound on communication delay in the AAM system for the first time. Furthermore, we formulate a complex non-convex optimization problem that jointly minimizes the probabilistic delay bound and the propagation delay. To solve this problem efficiently, we propose a solution based on the Block Coordinate Descent (BCD) algorithm and Semidefinite Relaxation (SDR) method. In addition, we conduct a comprehensive analysis of how various factors impact regret and transmission rate, and explore the influence of varying load intensity and total delay on the probabilistic delay bound. Under identical meta-atom settings, simulation results indicate that the proposed method achieves a 51.47\% average improvement in transmission rate over the low-complexity Alternating Optimization (AO) algorithm.
This work provides a theoretical foundation and technical reference for future delay performance optimization and safety assurance in AAM systems.
\end{abstract}

\begin{IEEEkeywords}
Advanced air mobility, stacked intelligent metasurfaces, network calculus, semidefinite relaxation method.
\end{IEEEkeywords}

\IEEEpeerreviewmaketitle

\section{Introduction}
\IEEEPARstart{R}{ecently}, rapid urban population growth coupled with escalating transportation demands have exacerbated ground traffic congestion. Conventional transportation infrastructures face significant challenges in accommodating the dynamic mobility requirements of contemporary urban environments, simultaneously contributing to prolonged commute durations and heightened environmental degradation. In response, Advanced Air Mobility (AAM) has emerged as a transformative paradigm that exploits low-altitude airspace to deliver diverse transportation services, encompassing cargo logistics, passenger transit, and emergency response operations~\cite{9447255}. Industry forecasts project the commercial deployment of AAM technologies by 2028, with anticipated annual passenger throughput reaching approximately 440 million by 2050~\cite{9952882}.

%\IEEEPARstart {R}{ecently}, the rapid growth of urban populations and increasing  transportation demands have intensified ground traffic congestion. Traditional transportation systems not only struggle to meet the evolving mobility need of modern cities, but also contribute to extended commuting time and aggravated environmental pollution. Against this backdrop, Advanced Air Mobility (AAM) has emerged as a promising solution, leveraging low-altitude airspace to provide a wide range of transportation service, including cargo delivery, passenger transport, and emergency service~\cite{9447255,10356185,10803914}. AAM is projected to be commercially deployed by 2028 and is expected to serve approximately 440 million passengers annually by 2050~\cite{9952882}.

Typically, the realization of AAM relies on a suite of advanced technologies, including Unmanned Aerial Vehicles (UAVs)~\cite{10343091,10915325},
Electric Vertical Take-Off and Landing (eVTOL) aircraft~\cite{10749494,10388419}, and advanced air traffic management systems~\cite{10677063}. These technologies collectively facilitate a novel transportation paradigm that complements existing ground-based systems. By leveraging the underutilized low-altitude airspace beneath conventional aviation corridors, AAM not only alleviates ground traffic congestion and reduces travel times but also contributes to mitigating the environmental footprint associated with traditional transportation modalities.

%Typically, the realization of AAM is based on a variety of advanced technologies, such as Unmanned Aerial Vehicle (UAV)~\cite{10343091,10915325}, Electric Vertical Take-off and Landing (eVTOL) aircraft~\cite{10749494,10388419}, and advanced air traffic management systems~\cite{10677063}. These technologies collectively enable a new mode of transportation that complements existing ground transportation. By fully utilizing the low-altitude airspace beneath conventional aviation corridors, AAM not only helps to alleviate ground traffic congestion and reduce travel time , but also mitigates the environmental impact of traditional modes of transportation.

Similarly to ground Vehicle-to-Vehicle (V2V) communication systems, eVTOL aircraft necessitate a dedicated aerial communication framework, termed eVTOL-to-eVTOL (eV2V) communication~\cite{10230034}, to support cooperative perception and navigation. This communication system is engineered to satisfy the stringent requirements of inter-eVTOL data exchange, networking, and onboard computational capabilities. By facilitating real-time sharing of critical flight parameters among eVTOLs, the system effectively mitigates the risk of mid-air collisions. Consequently, the deployment of eV2V communication enhances airspace utilization efficiency while ensuring the operational safety of AAM networks.

%Similarly to Vehicle-to-Vehicle communication, eVTOL aircraft requires an efficient aerial communication system, called eVTOL-to-eVTOL (eV2V) communication~\cite{10230034}, to enable cooperative perception and navigation. This communication architecture is designed to meet the demands of inter-eVTOL data exchange, networking, and onboard computation. By allowing eVTOL aircraft to share critical flight parameters in real time, the system can significantly reduce the risk of in-flight collisions. As a result, AAM can achieve improved airspace utilization while ensuring operational safety.

Despite its potential, eV2V communication in AAM systems confronts substantial challenges. In 6G-enabled AAM networks, achieving ultra-low latency, on the order of 1 ms, is essential to guarantee timely and safe operations~\cite{9952749}. Considering the high mobility of low-altitude aerial vehicles, even marginal communication delays can lead to slower response times and elevate the risk of accidents. Consequently, the development of eV2V communication technologies that offer both low latency and high reliability~\cite{10444151,10539166} is critical to the safe and efficient functioning of AAM systems.

%However, eV2V communication in AAM systems still faces significant challenges. In 6G-enabled AAM communication, ultra-low latency requirement—as low as 1~ms—is critical for timely and safe operation~\cite{9952749}. Given the high-rate mobility of low-altitude aerial vehicles, even minimal communication latency may result in delayed responses and increased risk of accidents. Therefore, developing low-latency and highly reliable eV2V communication technologies~\cite{10444151,10539166} is imperative to ensure safe and efficient operation of AAM systems.

Compared to conventional Multiple Input Multiple Output (MIMO) antenna systems, Holographic MIMO (HMIMO) systems exhibit enhanced performance in terms of both energy efficiency and transmission capacity~\cite{10163760,10301687,10232975}. 
%~\cite{9136592,10163760,10501134,10301687,10232975}
From a hardware perspective, holographic MIMO surfaces (HMI-MOS) feature a thinner profile and reduced power consumption. HMIMO systems combine the multi-antenna benefits of traditional MIMO with advanced signal processing techniques, enabling significantly higher data transmission rates while lowering energy usage~\cite{9848831}, 
thereby contributing to reduced communication latency. By employing HMIMO communication, eVTOL aircraft can sustain high-rate data transmission with decreased battery drain, ultimately extending flight range and enhancing overall endurance.

%Compared to conventional Multiple Input Multiple Output (MIMO) antenna systems, Holographic Multiple Input Multiple Output (HMIMO) systems demonstrate superior performance in both energy efficiency and transmission capability~\cite{9136592,10163760,10501134,10301687,10232975}. From a hardware perspective, holographic MIMO surfaces (HMI-MOS) are characterized by a thinner form factor and lower power consumption. HMIMO systems integrate the multi-antenna advantages of traditional MIMO systems with advanced signal processing techniques, enabling significantly higher data transmission rate while reducing power consumption~\cite{10603280,9848831}, thus contributing to lower communication latency. By leveraging HMIMO communication systems, eVTOL aircraft can maintain high-rate data transmission with reduced battery consumption, ultimately extending the range of the flight and improving overall flight endurance.

Conventional HMIMO systems are generally composed of a single-layer metasurface and can be classified into two categories: active HMIMO systems, exemplified by Large Intelligent Surfaces (LIS)~\cite{8264743}, and passive HMIMO systems, represented by Reconfigurable Intelligent Surfaces (RIS)~\cite{8741198,9086766}. Nevertheless, these architectures are inherently limited in their degrees of freedom for optimization, which restricts their adaptability and performance in dynamic environments.

%Conventional HMIMO systems are typically composed of a single-layer metasurface and can be categorized into two types: active HMIMO, such as Large Intelligent Surfaces (LIS)~\cite{8264743}, and passive HMIMO, such as Reconfigurable Intelligent Surfaces (RIS)~\cite{8741198,9086766}. However, these architectures are limited in terms of degree of freedom of optimization, which constrains their adaptability and performance in dynamic environment.

To overcome these limitations, the concept of Stacked Intelligent Metasurfaces (SIM) has been proposed~\cite{10158690,10515204}. By employing a multilayer stacked metasurface architecture, SIM significantly enlarges the optimization space, improves transmission efficiency, and enables signal processing directly in the analog domain. However, investigations into the delay performance of SIM in low-altitude scenarios remain limited, posing challenges to its practical implementation in AAM systems. A primary challenge lies in guaranteeing high reliability while effectively controlling the maximum propagation delay.

%To address this limitation, the concept of Stacked Intelligent Metasurfaces (SIM) has been introduced~\cite{10158690}. Using a multilayer stacked metasurfaces structure, SIM offers expanded optimization space and enhanced transmission performance. Nevertheless, studies on SIM communication performance in low-altitude scenarios remain scarce, which hinders its practical deployment in AAM. Among the key challenge is to ensure high reliability and control the maximum transmission delay.

%stochastic network calculus enables the characterization of probabilistic delay bounds~\cite{10039312,7011339,10437275,5593629}, 

This paper presents a systematic analysis of the delay performance of SIM in low-altitude scenarios by leveraging network calculus, a rigorous mathematical framework for performance evaluation. The objective is to provide theoretical insights to support the deployment of SIM within AAM systems. Unlike deterministic network calculus, which only accounts for fixed upper bounds on delay~\cite{9374444}, stochastic network calculus enables the characterization of probabilistic delay bounds~\cite{10039312,5593629}, making it more suitable for complex multi-hop networks with dynamic variations~\cite{6733260}. This approach is particularly critical for collaborative eVTOL scenarios with stringent real-time requirements, as it not only guarantees reliable performance but also facilitates the effective analysis of communication systems subject to high uncertainty and dynamic behaviors.

%This paper conducts a systematic analysis of the communication performance of SIM in low-altitude scenarios using network calculus, a mathematical tool for performance evaluation. This analysis aims to provide theoretical guidance for the application of SIM in AAM. Compared to deterministic network calculus, which can only handle fixed upper bound on delay~\cite{9374444}, stochastic network calculus offers the capability to characterize probabilistic delay bound~\cite{10039312,7011339,10437275,5593629}, making it more suitable for complex multi-hop network with dynamic variations~\cite{6733260}. This method is particularly important for collaborative eVTOL scenarios with stringent real-time requirement, as it not only ensures reliable performance guarantee but also enables effective analysis of communication systems with high uncertainty and dynamic behavior.

Based on the aforementioned performance analysis, this paper proposes a framework leveraging the Block Coordinate Descent (BCD) algorithm~\cite{5946716} to decompose the complex optimization problem. To optimize the configuration of SIM phase shifts, a Semidefinite Relaxation (SDR) based strategy~\cite{8811733} is employed, aiming to jointly maximize the transmission rate and minimize the probabilistic delay bound in SIM-enabled communication systems.
The main contributions of this paper are summarized as follows:
%Based on the aforementioned performance analysis, this paper proposes a solution framework based on the Block Coordinate Descent (BCD) algorithm~\cite{5946716} to decompose the optimization problem. An optimization strategy based on the Semidefinite Relaxation (SDR) method~\cite{8811733} is employed to optimize the configuration of SIM phase shifts, aiming to simultaneously maximize the transmission rate and minimize the probabilistic delay bound of SIM-based communication.

\begin{itemize}
\item We propose a network calculus method-based to SIM-based delay performance analysis for the first time, providing a novel perspective for optimizing low-altitude communication propagation delay by rigorously analyzing the probabilistic upper bound. This study not only addresses a theoretical gap in the delay performance analysis of SIM communication systems but also develops an optimization framework that balances the maximization of transmission rate and the minimization of propagation delay, thereby improving the overall system efficiency and reliability. The proposed framework offers an innovative approach for deploying SIM technology in low-altitude environments, further enhancing its practical value and significance within intelligent low-altitude networks.

%We proposed the network calculus method to SIM communication for the first time, providing a new perspective for optimizing low-altitude communication by deep-analyzing the upper bound of probability exceeding specific delay. This study not only fills the theoretical gap in SIM communication performance analysis, but also proposes an optimization framework aimed at balancing the maximization of transmission rate and minimization of transmission delay, thereby enhancing the overall efficiency and reliability of the system. This framework offers an innovative solution for the application of SIM technology in low-altitude scenarios, further expanding its value and practical significance in low-altitude intelligent network.

\item We develope an SDR-based approach to solve the optimization problem, leveraging it to enhance channel gain through precise optimization of the SIM phase shifts, thereby effectively improving the transmission rate. By optimizing the configuration of SIM phase shifts, the SDR method not only reduces propagation delay but also increases the transmission rate, offering an effective solution to improve the overall performance of low-altitude communication systems. This approach establishes a solid foundation for efficient, low-latency communications, and facilitates the advancement of SIM technology in practical deployments.

%We designed the SDR method to solve the optimization problem and utilize it to enhance channel gain by precisely optimizing the phase shifts of SIM, thereby effectively improving the transmission rate. By optimizing the configuration of SIM phase shifts, the SDR method not only reduces delay but also boosts transmission rate, providing an effective solution to enhance the overall performance of low-altitude communication systems. This approach lays a solid foundation for efficient, low-latency communication and promotes further development of SIM technology in practical applications.

\item Simulation results demonstrate that the proposed BCD algorithm outperforms the low-complexity Alternating Optimization (AO) algorithm in optimizing the transmission rate. Furthermore, we systematically analyze the impact of various parameters on transmission rate and system regret, and investigate the influence of system load and total delay on the probabilistic delay bound.
%We developed simulation results that demonstrate that the proposed BCD algorithm outperforms the low-complexity alternating optimization (AO) algorithm in terms of transmission rate optimization. Furthermore, we systematically analyze the impact of various factors on \(v_{\text{data}}\) and system regret, and investigate how system load and total delay influence the probabilistic delay bound.

\end{itemize}

The remainder of this paper is organized as follows. Section~\ref{sec:Related Work} reviews the related literature. Section~\ref{sec:System Model} describes the system model. In Section~\ref{sec:SNC-assisted Probabilistic Delay Bound Analysis}, we derive the end-to-end probabilistic delay bound based on stochastic network calculus theory. Section~\ref{sec:SIM Performance Analysis} defines the transmission rate. Section~\ref{sec:Problem Formulation and Optimization} formulates the optimization problem and introduces a BCD-based solution method. Section~\ref{sec:Performance Evaluation} presents the simulation results, and Section~\ref{sec:Conclusion} concludes the paper.

% \textit{Notations}: The operators \((\cdot)^T\) and \((\cdot)^H\) denote the transpose and the Hermitian transpose of a matrix, respectively. The expectation operation is represented by \(\mathbb{E}(\cdot)\). The notation \(\operatorname{diag}(\boldsymbol{\Lambda})\) represents a diagonal matrix whose diagonal elements are given by the vector \(\boldsymbol{\Lambda}\). The space of \(x \times y\) complex matrices is represented by \(\mathbb{C}^{x \times y}\). The identity matrix is denoted by \(\mathbf{I}\), and the trace of a matrix is represented as \(\operatorname{tr}(\cdot)\).

%The logarithmic function with base 2 is denoted as \(\log_2(\cdot)\). 
% The square root of a square matrix \(\mathbf{S}\) is denoted as \(\mathbf{S}^{1/2}\). The sinc function is defined as \(\operatorname{sinc}(x) = \frac{\sin(\pi x)}{\pi x}\). 

\section{Related Work}
This section reviews the  relevant literature, which can be broadly categorized into three key areas: (i) eVTOL transportation systems, (ii) the fundamentals of stacked intelligent metasurfaces, and (iii) the basic principles of network calculus. A brief overview of related studies in each of these categories is presented below.

\label{sec:Related Work}
\subsection{eVTOL Transportation System}

The ongoing acceleration of global urbanization has led to a continuous increase in passenger transportation demand, resulting in intensified traffic congestion on the ground due to overwhelming traffic volume. There is an urgent need for efficient solutions to overcome the limitations of the ground transportation system. One of the most promising approaches is the development of an aerial transportation system based on eVTOL aircraft. eVTOL-based transportation systems are designed to offer a safe, efficient and accessible mode of travel within large urban agglomerations~\cite{9925754}.

A typical eVTOL system consists of three core components: the vehicle platform (eVTOL aircraft), the airspace infrastructure (designated air corridors), and the support system (high-reliability communication network). The service scenarios for the eVTOL systems include primarily: air taxis, air metro, last-mile delivery services, passenger transportation, and emergency service~\cite{10356185,10803914}.

\subsection{Stacked Intelligent Metasurfaces}

In wireless communication systems, traditional RIS generally adopts a single layer metasurface architecture. This design presents significant limitations in beamforming degree of freedom~\cite{10767193} and multi-user interference suppression~\cite{10515204}. To overcome this technical bottleneck, SIM achieves performance leap through an innovative multi-layer stacked metasurfaces structure~\cite{10158690}.

A typical SIM system consists of several metasurface layers, with each layer containing multiple meta-atoms. These layers are independently controlled by a Field Programmable Gate Array (FPGA) controller. This hierarchical structure provides SIM with multi-dimensional degree of freedom, enabling layer-by-layer dynamic optimization of electromagnetic parameters such as phase shifts and amplitude. As a result, SIM can precisely reconstruct the spatial electromagnetic field, including the direction of wave propagation, reflection coefficients, polarization characteristics and so on.

%including beamforming~\cite{10279173}, SISO/MISO system enhancement~\cite{10767193}, channel estimation optimization~\cite{10445164}, satellite communication~\cite{10445200}, and non-cellular network~\cite{10535263},direction-of-arrival (DOA) estimation~\cite{an2024two}.

Current research demonstrates that SIM technology exhibits significant advantages in various areas, including beamforming~\cite{10279173}, SISO/MISO system enhancement~\cite{10767193}, channel estimation optimization~\cite{10445164}, satellite communication~\cite{10445200}, eVTOL communication~\cite{xiong2025digitaltwinbasedsimcommunication}, Direction-of-Arrival (DOA) estimation~\cite{an2024two}. However, its application in AAM remains unexplored. In particular, for communication scenarios between eVTOLs, as illustrated in Fig.~\ref{scene}, there is a noticeable gap in the published research on the performance optimization of the SIM-based eV2V communication system. This technological gap needs to be addressed urgently.

\subsection{Stochastic Network Calculus}
% Network calculus is a mathematically grounded theoretical framework used to analyze performance parameters in network systems, such as delay, backlog. It is generally categorized into two main branches: Deterministic Network Calculus (DNC) and Stochastic Network Calculus (SNC). 
%DNC is characterized by its ability to derive deterministic upper bound on service performance and provide guaranteed Quality of Service (QoS). It models traffic input through arrival curves and characterizes the service capability of network nodes via service curves, enabling the computation of worst-case delay ~\cite{9374444}. In contrast, 
Stochastic Network Calculus (SNC) is capable of handling stochastic traffic variations in network and utilizes probabilistic tools to provide 
performance bound~\cite{10039312}. These capability makes network calculus highly applicable in various domains, including optimization of communication network performance~\cite{10288334}, Industrial Internet of Things (IIoT)~\cite{5554709}, and low-altitude communication network~\cite{10706082}. 
%Furthermore, network calculus supports compositional performance analysis in distributed network, making it particularly suitable for optimizing complex systems operating in dynamic environment~\cite{6733260}.

The following introduces some common concepts in network calculus.  
Define \( F \) as a non-negative, non-decreasing function, and define \( F^c \) as a non-negative, non-increasing function.
% \begin{equation}
% \begin{split}
% \begin{aligned}
% \mathrm{F}
% &=\left\{\mathrm{f}(\cdot): \forall 0 \leq \mathrm{x}_1 \leq \mathrm{x}_2, 0 \leq \mathrm{f}\left(\mathrm{x}_1\right) \leq \mathrm{f}\left(\mathrm{x}_2\right)\right\}, \\
% \mathrm{F^{c}}
% &=\left\{\mathrm{f}(\cdot): \forall 0 \leq \mathrm{x}_1 \leq \mathrm{x}_2, 0 \leq \mathrm{f}\left(\mathrm{x}_2\right)\leq \mathrm{f}\left(\mathrm{x}_1\right) \right\}.
% \end{aligned}
% \end{split}
% \label{_w_ripuewaifgw}
% \end{equation}
For any random variable \( X \), its distribution function is defined as:
\(F_X(x) = P(X \leq x)\) The complementary distribution function is defined as: \(\bar{F}_X(x) = P(X > x)\). If both \( X \) and \( Y \) are random variables, for all \( x \geq 0 \), if \(P(X > x) \leq f(x), \quad P(Y > x) \leq g(x)\) then we have~\cite{7011339}: \(P(X + Y > x) \leq f(x) \otimes g(x)\), where \(\otimes\) denotes the min-plus convolution, which is defined as follows: $f \otimes g(t)=\inf _{0<\tau \leq t}\{f(\tau)+g(t-\tau)\}$.

Let $A(t)$ denotes the arrival process, which represents the cumulative number of bits that have arrived during the time interval $[0, t]$. Similarly, let \( S(t) \) denotes the service process, defined as the cumulative number of bits that have been served in the interval \([0, t]\). The departure process, denoted by \( A^*(t) \), refers to the cumulative number of bits that have departed within the same interval. For any \( 0 \leq \tau \leq t \), the arrivals, services and departures in the subinterval \([ \tau, t ]\) are given by: $A(\tau, t) = A(t) - A(\tau)$, $S(\tau, t) = S(t) - S(\tau) $, $A^*(\tau, t) = A^*(t) - A^*(\tau)$.

For stochastic network calculus~\cite{7011339}, a data flow is said to conform to a stochastic arrival curve \( \alpha \in F \) with a bounding function \( f \in F^c \), denoted as \( A \sim \langle f, \alpha \rangle \), if for all \( t \geq 0 \) and \( x \geq 0 \), the following inequality holds:
\begin{equation}
\begin{split}
\begin{aligned}
P \left\{ \sup_{0 \leq \tau \leq t} \left[ A(\tau, t) - \alpha(t - \tau) \right] > x \right\} \leq f(x).
\end{aligned}
\end{split}
\label{hrtuktykmfumyu}
\end{equation}

Similarly, a server is said to provide a stochastic service curve \( \beta \in F \) with a bounding function \( g \in F^c \), denoted as \( S \sim \langle g, \beta \rangle \), if for all \( 0 \leq s \leq t \) and \( x \geq 0 \), it satisfies: $P \left\{ A \otimes \beta(t) - A^*(t) > x \right\} \leq g(x)$.
For a FIFO (First In First Out) system, the delay \( D \) at time \( t \geq 0 \) is given by: $D = \inf \left\{ \tau \geq 0 \mid A(t) \leq A^*(t + \tau) \right\}$.
% \begin{equation}
% \begin{split}
% \begin{aligned}
% D = \inf \left\{ \tau \geq 0 \mid A(t) \leq A^*(t + \tau) \right\}.
% \end{aligned}
% \end{split}
% \label{yuyurtdrhersg}
% \end{equation}

\section{System Model}
\label{sec:System Model}
In low-altitude intelligent network scenarios, SIM technology offers  significant advantages, including a lightweight structure and ease of integration, enabling seamless incorporation into eVTOL aircraft. We consider a bidirectional air-to-air (A2A) communication scenario where two eVTOLs, denoted as ${\rm eVTOL_A}$ and ${\rm eVTOL_B}$, are each equipped with both a transmitter SIM (TX-SIM) and a receiver SIM (RX-SIM), as illustrated in Fig.~\ref{scene}. This configuration enables direct A2A communication between the two aircraft via their respective SIM modules. Through the SIM-enabled A2A communication link, critical real-time information, including flight status, position, altitude, and velocity, can be exchanged to support safe and efficient aerial operations.

%In the context of low-altitude intelligent networks, SIM technology offers significant advantages, including a lightweight structure and ease of integration, facilitating seamless incorporation into eVTOL aircraft. Specifically, $eVTOL_A$ is equipped with a transmitter SIM(TX-SIM), while $eVTOL_B$ integrates a receiver SIM (RX-SIM), with detailed designs available in ~\cite{10158690}. As illustrated in the scenario in Fig.~\ref{scene}, the two eVTOLs can exchange critical flight parameters such as flight status, position, altitude, and velocity in real time via an aerial communication link. This performs collaborative aerial sensing and autonomous obstacle avoidance, ensuring the safety of low-altitude air traffic. This SIM-assisted eV2V communication not only enhances link performance (e.g., data rate), but also provides crucial support for building a safe, efficient and intelligent low-altitude air traffic system.
\begin{figure}[h]
\centering
     \includegraphics[width=0.40\textwidth]{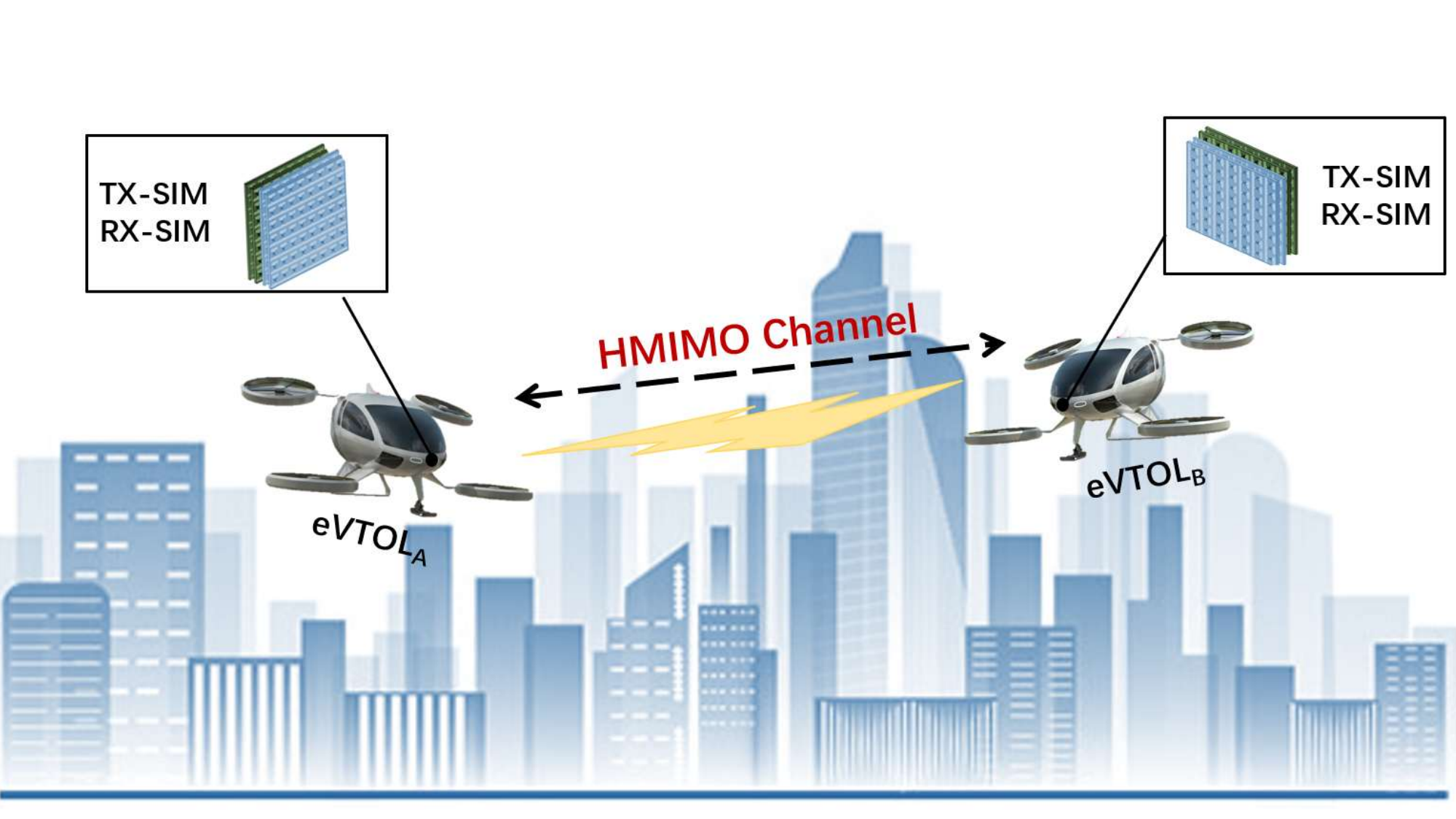} 
     \caption{Illustration of the SIM-Aided A2A Communication between eVTOLs.} 
\label{scene}
\end{figure}
Considering data freshness as a key requirement for flight safety, the communication delay during data packet transmission must be constrained within a reasonable range to avoid potential collision risks. Therefore, in this work, we aim to derive a probabilistic delay bound, $P\{D > T\}$, defined as the probability that the packet delay $D$ exceeds a given threshold $T$, based on the Stochastic Network Calculus theory.

%Considering the stringent requirements of flight safety, the communication delay during data packets transmission must be constrained within a reasonable range to avoid awareness latency caused by communication delay, which may result in potential collision risk. Therefore, based on the theory of {\color{red} SNC}, this section derives the probabilistic delay bound \(P\{D > t\} \),which means the probability that the packet delay \( D \) exceeds a threshold \( t \), for the transmission process from $eVTOL_A$ to $eVTOL_B$. This analysis provides theoretical support for the safe operation of low-altitude intelligent traffic systems, while ensuring high reliability and real-time performance of the communication system.

In the considered system, we make the following red three assumptions: (i) prior to transmission, data packets are temporarily stored in a buffer equipped on $\rm eVTOL_A$, which operates under a FIFO principle;  (ii) the buffer has sufficient capacity, and issues such as packet loss or retransmission are beyond the scope of this work; and (iii) the service rate is greater than the data arrival rate, ensuring that the queue length remains bounded.

%To more effectively analyze the impact of the SIM system on the delay in eV2V communication, the following system assumptions are made: After transmitting \( S \) data streams from the transmit antenna in the TX-SIM on $eVTOL_A$, the data packets are transmitted in the TX-SIM, and then the data packets are temporarily stored in a buffer, which follows the FIFO principle. The buffer is assumed to have sufficient capacity and issues such as packet loss or retransmission are not considered.Subsequently, the packets are transmitted through the TX-SIM via the HMIMO channel, experience transmission delay \( D_3 \), and are finally received by the RX-SIM module on $eVTOL_B$ before being forwarded to the receive antenna. This process is illustrated in Fig.~\ref{transmission process}.

%Consequently, packets are transmitted through the TX-SIM via the HMIMO channel, experience a transmission delay $D_3$, are received by the RX-SIM module on eVTOL${\rm B}$, and are finally forwarded to the receiving antenna, as illustrated in Fig.~\ref{transmission process}.

\begin{figure}[h]
\centering
\includegraphics[width=0.45\textwidth]{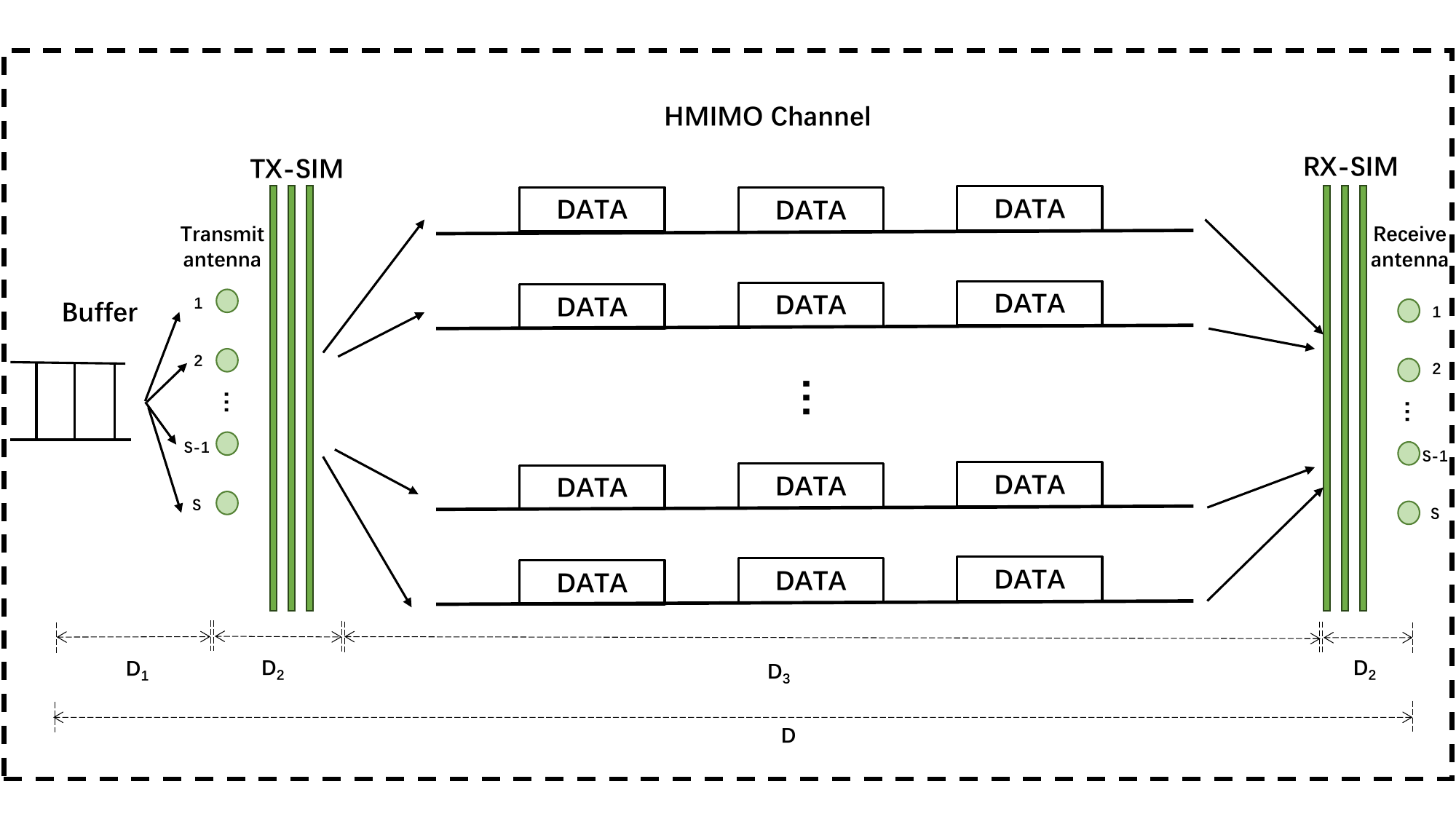} 
\caption{Illustration of the End-to-End Packet Delivery Mechanism for eV2V Communication.} 
\label{transmission process}
\end{figure}

As illustrated in Fig.~\ref{transmission process}, the total packet delay $D$ consists of the queuing time in the buffer, delays in both the TX- and RX-SIM modules, and the propagation delay over the HMIMO channel. Specifically,
\begin{itemize}
\item $D_1$: \textbf{Queuing delay}, denoting the time that data packets wait in the buffer before transmission.
\item $D_2$: \textbf{Transmission delay}, accounting for the transmission time from the $s$-th source to the first metasurface layer of the TX-SIM, as well as the internal transmission delay within the RX-SIM. This delay is considered constant in this work.
\item $D_3$: \textbf{Propagation delay}, representing the time required for data packets to be transmitted across the HMIMO channel from the output metasurface of the TX-SIM to the input metasurface of the RX-SIM.
\end{itemize}
% In this communication procedure, the processing delay of the data packets is neglected. As such, the total delay of packets  $D$ is decomposed into the following three components:
% \begin{itemize}
%     \item \( D_1 \): \textbf{Queuing delay}, representing the time the data packets spend waiting in the buffer before being transmitted.
%     \item \( D_2 \): \textbf{radio transmission delay}, representing fixed overhead. It accounts for the propagation time between the \( s \)th source and the first metasurface layer of the TX-SIM, as well as the internal propagation delay within the RX-SIM.
%     \item \( D_3 \): \textbf{Transmission delay}, representing the time taken for data packets to be transmitted through the HMIMO channel from the output metasurface of the TX-SIM to the input metasurface of the RX-SIM.
% \end{itemize}

Based on the aforementioned assumptions and system model, we characterize the statistical properties of the total end-to-end delay, $D$, which can be expressed as the sum of the queuing delay, transmission delay, and propagation delay, $D = D_1 + D_2 + D_3$, and our goal is to derive the probabilistic upper bound on the end-to-end delay, as well as to minimize both the violation probability and the overall delay. These are defined and derived in the following subsections.
 \begin{definition}[End-to-end Delay of Packets] 
The probabilistic upper bound on the end-to-end delay is given by
\begin{align}
\bar{D} = P \{D_1 + D_2 + D_3 > T\} ,
\end{align}
where $T$ denotes the delay threshold. Our objective is to minimize both $\bar{D}$ and $T$.
\end{definition}

\section{SNC-assisted Probabilistic Delay Bound Analysis}
\label{sec:SNC-assisted Probabilistic Delay Bound Analysis}
In this section, we analyze the upper bound of the total end-to-end delay in the proposed system. We begin by examining the queuing delay in the buffer, which is expressed as a function of the queuing delay. Subsequently, we analyze the propagation delay itself and formulate an optimization problem to minimize it. This, in turn, corresponds to minimize the upper bound of the total packet delay in the considered system.

\subsection{Queueing Delay}
In this subsection, we analyze the queuing delay of the data in the buffer. Note that only the queuing delay on the TX-SIM side is considered in this analysis. The upper bound of $D_1$ is given in the following lemma.
\begin{lemma}[Upper Probabilistic Bound of Queueing Delay] 
\label{theorm_barD1}
The queueing upper probabilistic bound is given by 
\begin{align}
 \bar{D}_1 = \epsilon(v_{\text{data}} B t_b),
\end{align}
% \begin{align}
%  \bar{D}_1 = P\{D_1 > t_b\} \leq \epsilon(v_{\text{data}} B t_b)=\bar{D}_1,
% \end{align}
where  $\beta_{\text{data}}(t_b) $ is the service curve, $t_b$ is the expected queuing delay, and $\epsilon(\cdot)$ is the bounding function, in case of Poisson arrival process $\epsilon(x) = e^{-\mu x} $, and $ \mu > 0 $ is a parameter that minimizes the bounding function  $\epsilon(x)$, and
\begin{align}
\beta_{\text{data}}(t) = v_{\text{data}}  B t,
\end{align}
in which $B$ denotes the system bandwidth and $v_{data}$ is the transmission rate in the HMIMO channel.
\end{lemma}
\begin{proof}
    See Appendix \ref{app_barD1}.
\end{proof}

Recall that we assume the arrival process $A(t)$ follows a Poisson distribution, and the packet size follows an exponential distribution with mean $l_d$. Let $\delta_d(t)$ denote the number of packets arriving at the HMIMO channel at time $t$, and let $\delta(t)$ denote the number of packets generated from $S$ sources. In a stable network, we have $\delta_d(t) = \delta(t - D_2)$.

The stochastic arrival curve can then be expressed as~\cite{jiang2010note}:
\begin{equation}
\alpha(t) =\frac{1}{\mu } \log \mathbb{E}[e^{\mu A(t)}]  = \frac{\delta_d l_d t}{1 - \mu l_d},
\end{equation}
and consequently, according to the stability condition \( \alpha(t) \leq \beta(t) \), the optimal value of \( \mu \) is given by:
\begin{equation}
\mu = \frac{v_{\text{data}} B - \delta_d l_d}{v_{\text{data}} B l_d}.
\label{eq:mu}
\end{equation}

\begin{lemma}[Queueing Probabilistic Delay Bound in Case of Poisson Arrival Process]
Combining the bounding function $\epsilon(x)=e^{-\mu x}$ and ~\eqref{eq:mu}, the probabilistic delay bound of \( D_1 \) in case of Poisson arrival process is given by:
\begin{equation}
\bar{D}_1(t_b) = e^{-\mu v_{\text{data}} B t_b} = e^{-\frac{v_{\text{data}} B - \delta_d l_d}{l_d} t_b}.
\label{eq_pdfD1}
\end{equation}
\end{lemma}
\subsection{Transmission delay}
Assume that there are $S$ data streams in total. The TX-SIM consists of $L$ metasurface layers, each containing $M$ meta-atoms, where $M \geq S$. Similarly, the RX-SIM is composed of $K$ metasurface layers, each with $N$ meta-atoms, and $N \geq S$. Under these conditions, the transmission delay $D_2$ can be expressed as\begin{equation}
\begin{split}
\begin{aligned}
% D_2 = \max(\frac{d_{s,\text{TX}}}{v_{\text{TX}}} + \frac{d_{A-\text{TX}}}{v_{\text{TX}}} + \frac{d_{B-\text{RX}}}{v_{\text{RX}}} + \frac{d_{\text{RX},s}}{v_{\text{RX}}} )\\
D_2 = \max(\frac{d_{s,\text{TX}} + d_{A-\text{TX}}}{v_{\text{TX}}} + \frac{d_{B-\text{RX}} + d_{\text{RX},s}}{v_{\text{RX}}}),
\end{aligned}
\end{split}
\label{eq:D_2}
\end{equation}
where $v_{\text{TX}}$ and $v_{\text{RX}}$ denote the electromagnetic wave propagation speeds within the TX-SIM and RX-SIM, respectively. The term $d_{s,\text{TX}}$ represents the distance between the $s$-th data source and the $m$-th meta-atom in the input metasurface layer of the TX-SIM, while $d_{\text{RX},s}$ denotes the distance between the $n$-th meta-atom in the output metasurface layer of the RX-SIM and the $s$-th data source. The total intra-device propagation distances within the TX-SIM and RX-SIM are denoted by $d_{A\text{-}TX}$ and $d_{B\text{-}RX}$, respectively. The values of $d_{s,\text{TX}}$ and $d_{\text{RX},s}$ can be calculated using the following formula
\begin{equation}
\begin{split}
\begin{aligned}
a_{TX} =  \left( m_z - \frac{m_{\text{max}}^r + 1}{2} \right) r_{\text{tx}} - \left( s - \frac{S + 1}{2} \right) \frac{\lambda}{2},
\end{aligned}
\end{split}
\label{eq:a_TX}
\end{equation}
\begin{equation}
\begin{split}
\begin{aligned}
a_{RX} =   \left( n_z - \frac{n_{max}^r + 1}{2} \right) r_{rx} - \left( s - \frac{S + 1}{2} \right) \frac{\lambda}{2},
\end{aligned}
\end{split}
\label{eq:a_RX}
\end{equation}
\begin{equation}
\begin{split}
\begin{aligned}
b_{TX} = m_x - \frac{m_{max}^r + 1}{2}r_{tx}, \ b_{RX} = n_x - \frac{n_{max}^r + 1}{2}r_{rx},
\end{aligned}
\end{split}
\label{eq:b_TX}
\end{equation}
\begin{equation}
\begin{split}
\begin{aligned}
d_{s,\text{TX}} = \sqrt{ a_{TX}^2 
+ b_{TX} ^2 + d_t^2 }, \ d_{\text{RX},s} = \sqrt{a_{RX}^2 
+ b_{RX}^2 + d_r^2 },
\end{aligned}
\end{split}
\label{eq:d_s,TX}
\end{equation}

\noindent where $m_z = \lceil \frac{m}{m_{\max}^r} \rceil$ and $n_z = \lceil \frac{n}{n_{\max}^r} \rceil$ represent the z-axis indices of the $m$-th and $n$-th meta-atoms in the TX-SIM and RX-SIM, respectively. Similarly, let $m_x = \operatorname{mod}(m-1, m_{\max}^r) + 1$ and $n_x = \operatorname{mod}(n-1, n_{\max}^r) + 1$ represent their respective x-axis indices, $m_{\max}^r$ denotes the maximum number of meta-atoms per row on each metasurface layer of the TX-SIM, while $n_{\max}^r$ denotes the same for the RX-SIM. The parameters $r_{\text{tx}}$ and $r_{\text{rx}}$ represent the distances between adjacent meta-atoms on each metasurface layer of the TX-SIM and RX-SIM, respectively. Let $\lambda$ denote the signal wavelength. The vertical spacing between metasurface layers in the TX-SIM is given by $d_t = \frac{D_t}{L}$, where $D_t$ is the total thickness of the TX-SIM and $L$ is the number of metasurface layers. Similarly, $d_r = \frac{D_r}{K}$ denotes the vertical spacing between metasurface layers in the RX-SIM, where $D_r$ is the total thickness and $K$ is the number of layers.

Next, we calculate $d_{\text{A-TX}}$ and $d_{\text{B-RX}}$, where $d_{\text{A-TX}}$ denotes the transmission distance of the $s$-th data stream within the TX-SIM of $eVTOL_A$, and $d_{\text{B-RX}}$ denotes the transmission distance of the $s$-th data stream within the RX-SIM of $eVTOL_B$. The corresponding formulas are given as follows
\begin{equation}
\begin{aligned}
d_{A-RX} =\sum_{l=2}^{\mathrm{L}} d_{A-R X}^{\mathrm{l}}, \ d_{B-RX} =\sum_{k=2}^{\mathrm{K}} d_{B-R X}^{\mathrm{k}},
\end{aligned}
\label{_w_ripuewaigw}
\end{equation}
\begin{equation}
\mathrm{~d}_{\mathrm{A}-\mathrm{TX}}^{\mathrm{l}}=\sqrt{{\mathrm{r}_{\mathrm{m}_{\mathrm{A}}, \mathrm{~m}_{\mathrm{A}}^{\prime}}^l}^2+\mathrm{d}_{\mathrm{t}}^2} \quad  l \in [2,L],
\label{eq:d_B_RX_k}
\end{equation}
\begin{equation}
d_{B-R X}^{\mathrm{k}}=\sqrt{{\mathrm{r}_{\mathrm{n}_{\mathrm{B}}, \mathrm{~n}_{\mathrm{B}}^{\prime}}^k}^2+\mathrm{d}_{\mathrm{r}}^2} \quad k\in[2,K],
\label{_w_ripuewaigw}
\end{equation}
where $d_{A\text{-TX}}^{l}$ denotes the distance between the $m_A$-th meta-atom on the $(l{-}1)$-th layer and the $m_A^{\prime}$-th meta-atom on the $l$-th layer of the TX-SIM in $eVTOL_A$. Similarly, $d_{B\text{-RX}}^{k}$ represents the distance between the $n_B$-th meta-atom on the $(k{-}1)$-th layer and the $n_B^{\prime}$-th meta-atom on the $k$-th layer of the RX-SIM in $eVTOL_B$.
The term $\mathrm{r}_{m_A,m_A^{\prime}}^{1}$ denotes the horizontal spacing between the $m_A$-th and $m_A^{\prime}$-th meta-atoms on the same metasurface layer of $eVTOL_A$. Likewise, $\mathrm{r}_{n_B, n_B^{\prime}}^{k}$ denotes the spacing between the $n_B$-th and $n_B^{\prime}$-th meta-atoms on the same metasurface layer of $eVTOL_B$.
The expressions for $\mathrm{r}_{m_A, m_A^{\prime}}^{1}$ and $\mathrm{r}_{n_B,n_B^{\prime}}^{k}$ are given as follows:
\begin{equation}
% \begin{aligned}
\mathrm{r}_{\mathrm{m_A}, \mathrm{~m_A}^{\prime}}^1=\mathrm{r}_{\mathrm{tx}} \sqrt{\left(\mathrm{~m}_{\mathrm{z}}-\mathrm{m}_{\mathrm{z}}^{\prime}\right)^2+\left(\mathrm{m}_{\mathrm{x}}-\mathrm{m}_{\mathrm{x}}^{\prime}\right)^2},
% \end{aligned}
\label{eq:r_mm_l}
\end{equation}
\begin{equation}
% \begin{aligned}
\mathrm{r}_{\mathrm{n_B}, \mathrm{~n_B}^{\prime}}^k=\mathrm{r}_{\mathrm{rx}} \sqrt{\left(\mathrm{~n}_{\mathrm{z}}-\mathrm{n}_{\mathrm{z}}^{\prime}\right)^2+\left(\mathrm{n}_{\mathrm{x}}-\mathrm{n}_{\mathrm{x}}^{\prime}\right)^2}.   
% \end{aligned}
\label{eq:r_nn_k}
\end{equation}

Finally, we are ready to derive the upper bound which is a fixed constant of the transmission delay, as given in the following lemma.
\begin{lemma}[Upper Bound of Transmission Delay]
The upper bound of the transmission delay is given by:
\begin{equation}
% \begin{split}
\begin{aligned}
 % D_2 &= \max(\frac{d_{\text{TX}, s} + d_{A\text{-TX}}}{v_{\text{TX}}} + \frac{d_{B\text{-RX}} + d_{\text{RX}, s}}{v_{\text{RX}}}) \\
        D_2 &= \max(\frac{\sqrt{a_{TX}^2 
+ b_{TX} ^2 + d_t^2 } 
            + \sum_{l=2}^{L} \sqrt{r_{m_A, m_A'}^1 + d_t^2}}{v_{\text{TX}}} \\
        &\quad + \frac{\sqrt{a_{RX}^2 
+ b_{RX} ^2 + d_r^2 } 
            + \sum_{k=2}^{K} \sqrt{r_{n_B, n_B'}^2 + d_r^2}}{v_{\text{RX}}} ).
\end{aligned}
% \end{split}
\label{eq_pdfD2}
\end{equation}
\end{lemma}
\begin{proof}
The result in (\ref{eq_pdfD2}) is obtained by combining ~\eqref{eq:D_2}--\eqref{eq:r_nn_k}.
\end{proof}

\subsection{Propagation Delay}
The upper probabilistic bound of the propagation delay \( \bar{D}_3(t) \) for transmitting \( S \) packets in the HMIMO channel can be expressed in the following lemma.
\begin{lemma}[Upper Probabilistic Bound of Propagation Delay] The upper probabilistic bound of propagation delay is given by:
\begin{align}
P\{D_3>t_d\}\leq\bar{D}_3(t_d) .
\end{align}
\end{lemma}
\begin{proof}The upper bound of $D_3$ is computed as follows:
\begin{equation}
\begin{aligned}
P\{D_3 > t_d\}= P\left\{ \frac{S l_d}{v_{\text{data}} B} > t_d \right\} 
= P\left\{ l_d > \frac{v_{\text{data}} B t_d}{S} \right\}.
\end{aligned}
\end{equation}
Let \( f_L(l) \) be the probability density function (PDF) of \( l_d \), the bounding function of $t_d$ is
\begin{equation}
\begin{aligned}
\label{eq_pdfD3}
\bar{D}_3(t_d) 
&= \int_{\frac{v_{\text{data}} B t_d}{S}}^{\infty} f_L(l) \, dl =e^{-\frac{v_{\text{data}} B}{S l_d} t_d}.
\end{aligned}
\end{equation}
\end{proof}
    
Now that the upper probabilistic bounds of \( D_1 \) and \( D_3 \) have been derived, we proceed to present the total upper probabilistic bound of the end-to-end delay in the following theorem. To facilitate the derivation, we first introduce a prerequisite lemma.

\begin{lemma}
Let \( Z = \sum_{i=1}^N X_i \), where \( Z \) and \( X_i \) are random variables. Then, the complementary cumulative distribution function (CCDF) of \( Z \) satisfies the following inequality
\begin{equation}
\bar{F}_Z(x) \leq \bar{F}_{X_1} \otimes \bar{F}_{X_2} \otimes \cdots \otimes \bar{F}_{X_N}(x),
\end{equation}
where \( \bar{F}_Z(x) \) denotes the CCDF of \( Z \) with respect to the variable \( x \), and \( \otimes \) represents the min-plus convolution operator.
\label{lemma_CCDF}
\end{lemma}

Based on Lemma~\ref{lemma_CCDF} and the statistical properties of \( D_1 \), \( D_2 \), and \( D_3 \), we now establish the upper probabilistic bound of the total end-to-end delay as follows.

\begin{theorem}[Upper probabilistic bound of $D$]
According to Lemma~\ref{lemma_CCDF}, the upper probability bound of the total delay is 
\begin{equation}
\begin{split}
\begin{aligned}
& \bar{D} =\inf _{\mathrm{t}_1+\mathrm{t}_3=\mathrm{T}-\mathrm{D}_2}\left\{\mathrm{e^{-\frac{v_{\text{data}} B - \delta_d l_d}{l_d} t_b}+\mathrm{e^{-\frac{v_{\text{data}} B }{S l_d}t_d}}}\right\}.
\end{aligned}
\end{split}
\label{vuyresgdfsuuysdriv}
\end{equation} 
\end{theorem}
\begin{proof}
The upper bound $\bar{D}$ is computed as
\begin{equation}
\begin{split}
\begin{aligned}
& \bar{D} = P\{\mathrm{D}>\mathrm{t}\}=\mathrm{P}\left\{\mathrm{D}_1(\mathrm{t})+\mathrm{D}_3(\mathrm{t})>\mathrm{T}-\mathrm{D}_2\right\} \\& \leq \mathrm{f}_{\mathrm{D}_1}\left(\mathrm{t}_{\mathrm{b}}\right) \otimes \mathrm{f}_{\mathrm{D}_3}\left(\mathrm{t}_{\mathrm{d}}\right),
\end{aligned}
\end{split}
\label{vuyresgdfsuuysdriv}
\end{equation} 
proof completes by substituting the CCDF of $D_1$ and $D_3$, as given in (\ref{eq_pdfD1}) and (\ref{eq_pdfD3}), respectively.
\end{proof}

We now derive the upper probabilistic bound of the delays associated with buffering, radio propagation, and transmission, and express the upper probabilistic bound of the total delay as a function of $v_{\rm data}$ and $t_d$, which are key parameters related to SIM-assisted transmission. In the following section, we analyze the optimal values of $v_{\rm data}$ and $t_d$ that minimize the total delay  $T$ and the total end-to-end probabilistic bound.

\section{SIM-based Performance Optimization}
\label{sec:SIM Performance Analysis}
Following the derivation of the probabilistic upper bound on the total end-to-end delay $\bar{D}$, a critical aspect remains to be addressed: the determination of the data transmission rate within the HMIMO channel. This transmission rate, denoted by $v_{\text{data}}$, directly influences the propagation delay $t_d$ and serves as a fundamental parameter in the performance evaluation of the proposed communication system. In the subsequent analysis, we derive an analytical expression for $v_{\text{data}}$.

The transmission rate in the HMIMO channel can be characterized to provide a quantitative basis for subsequent delay analysis~\cite{10534211}. Let $\theta_m^l \in [0, 2\pi)$ denote the phase shift of the $m$-th meta-atom in the $l$-th metasurface layer of the TX-SIM, where $m \in [1, M]$ and $l \in [1, L]$. The corresponding transmission coefficient is given by $\phi_m^l = e^{j\theta_m^l}$. The phase shift vector of the $l$-th layer is defined as $\boldsymbol{\phi}^1 = \left[\phi_1^1, \phi_2^1, \ldots, \phi_M^1\right]^{\mathrm{T}} \in \mathbb{C}^{M \times 1},$ and the associated transmission coefficient matrix is expressed as
$\mathbf{\Phi}^1 = \mathrm{diag}[\Phi^1] \in \mathbb{C}^{M \times M}$.
Similarly, let $\xi_n^k \in [0, 2\pi)$ denote the phase shift of the $n$-th meta-atom in the $k$-th layer of the RX-SIM, where $n \in [1, N]$ and $k \in [1, K]$. The corresponding transmission coefficient is given by $\psi_n^k = e^{j\xi_n^k}$. The phase shift vector of the $k$-th layer is defined as $ \boldsymbol{\psi}^k = \begin{bmatrix}
    \psi_1^k, \psi_2^k, ..., \psi_N^k
\end{bmatrix}^T \in \mathbb{C}^{N \times 1}$ and the corresponding transmission coefficient matrix is $\mathbf{\Psi}^1 = \text{diag}[\boldsymbol{\psi}^k] \in \mathbb{C}^{N \times N}$.

According to~\cite{10158690}, the transmission coefficient from the $m_A'$-th meta-atom in the $(l{-}1)$-th layer to the $m_A$-th meta-atom in the $l$-th layer of the TX-SIM carried by $eVTOL_A$ is given by
\begin{equation}
\mathrm{w}_{\mathrm{m}, \mathrm{~m}^{\prime}}^1=\frac{\mathrm{C}_{\mathrm{t}} \cos \chi_{\mathrm{m}, \mathrm{~m}^{\prime}}^1}{d_{A-TX}^l}\left(\frac{1}{2 \pi d_{A-TX}^l}-\mathrm{j} \frac{1}{\lambda}\right) \mathrm{e}^{\frac{\mathrm{j} 2 \pi d_{A-TX}^l }{\lambda} }, 
\end{equation}
%1 \in \mathrm{~L}
where $d_{A\text{-TX}}^l$ denotes the distance between the $m$-th meta-atom in the $l$-th metasurface layer and the $m'$-th meta-atom in the $(l{-}1)$-th layer, with $l \in [2, L]$. Here, $C_t$ is the effective area of each meta-atom in the TX-SIM, and $\chi_{m,, m'}^l$ is the angle between the wave propagation direction from the $(l{-}1)$-th layer and the normal vector of the $l$-th metasurface layer.

%Let $\mathbf{W}^l \in \mathbb{C}^{M \times M}$ denote the transmission coefficient matrix that characterizes the electromagnetic interaction between the $l$-th and $(l{-}1)$-th metasurface layers of the TX-SIM, for $l \in [2, L]$. Specifically, $\mathbf{W}^1 \in \mathbb{C}^{M \times S}$, then the effect of the TX-SIM is expressed as:
Let $\mathbf{W}^l \in \mathbb{C}^{M \times M}$ denote the transmission coefficient matrix that characterizes the electromagnetic interaction between the $l$-th and $(l{-}1)$-th metasurface layers of the TX-SIM, for $l \in [2, L]$. Specifically, $\mathbf{W}^1 \in \mathbb{C}^{M \times S}$ represents the transmission coefficient matrix mapping the $S$ input data streams to the first metasurface layer. Accordingly, the overall transformation implemented by the TX-SIM can be expressed as:
\begin{equation}
\mathbf{X} = \mathbf{\Phi}^L \mathbf{W}^L\mathbf{\Phi}^{L-1} \mathbf{W}^{L-1} \dots \mathbf{\Phi}^2 \mathbf{W}^2 \mathbf{\Phi}^1 \mathbf{W}^1 \in \mathbb{C}^{M \times S}.
\end{equation}

Similarly, for the RX-SIM, the transmission coefficient from the $n'$-th meta-atom in the $(k-1)$-th metasurface layer to the $n$-th meta-atom in the $k$-th layer is given by
%Similarly, for the RX-SIM, the transmission coefficient from the $n'$-th atom in the $(k-1)$-th metasurface layer of the transmission metasurface to the $n$-th atom in the $k$-th layer is expressed as
\begin{equation}
u_{n, n'}^k = \frac{C_r \cos \zeta_{n, n'}^k}{d_{\rm B-RX}^k}
\left( \frac{1}{2 \pi d_{\rm B-RX}^k} - \frac{1}{j \lambda} \right)
e^{j 2 \pi d_{B-RX}^k/\lambda},
\end{equation}
where $d_{B\text{-RX}}^k$ denotes the transmission distance between the $n'$-th meta-atom in the $(k{-}1)$-th metasurface layer and the $n$-th meta-atom in the $k$-th layer. The parameter $C_r$ represents the effective area of each unit meta-atom in the RX-SIM, and $\zeta_{n, n'}^k$ denotes the angle between the propagation direction from the $(k{-}1)$-th layer and the normal vector of the $k$-th metasurface layer.
%where $d_{B-RX}^k$ represents the transmission distance between the $n'$-th atom in the $(k-1)$-th metasurface layer and the $n$-th atom in the $k$-th layer, $C_r$ is the area of each unit atom in the RX-SIM, and $\zeta_{n, n'}^k$ is the angle between the propagation direction of the $(k-1)$-th metasurface layer and the normal direction.

%Let $\mathbf{U}^k$ denote the transmission coefficient matrix for metasurface layer $k$ and $k-1$, $k\in [2,K]$. Specifically, $\mathbf{U}^K \in \mathbb{C}^{S \times N}$ , then the effect of the RX-SIM is expressed as:
Let $\mathbf{U}^k \in \mathbb{C}^{N \times N}$ denote the transmission coefficient matrix representing the electromagnetic coupling between the $(k{-}1)$-th and $k$-th metasurface layers within the RX-SIM, for $k \in [2, K]$. Specifically, $\mathbf{U}^K \in \mathbb{C}^{S \times N}$ denotes the transmission coefficient matrix that maps the output of the final metasurface layer to the $S$ output data streams. Accordingly, the overall transformation induced by the RX-SIM can be expressed as:
\begin{equation}
\mathbf{Y} = \mathbf{U}^K \mathbf{\Psi}^K \mathbf{U}^{K-1} \mathbf{\Psi}^{K-1} \dots \mathbf{U}^2 \mathbf{\Psi}^2 \mathbf{U}^1 \mathbf{\Psi}^1  \in \mathbb{C}^{S \times N}.
\end{equation}

Consequently, the end-to-end channel matrix in HMIMO channel is given by
\begin{equation}
\mathbf{H} = \mathbf{Y} \mathbf{G} \mathbf{X},
\end{equation}
where $\mathbf{G} = \mathbf{R}_{\mathrm{RX}}^{1/2}\bar{\mathbf{G}}\mathbf{R}_{\mathrm{TX}}^{1/2} \in \mathbb{C}^{N \times M}$ is the channel matrix as described in~\cite{9714406}, $\bar{\mathbf{G}}$ represents an independent and identically distributed (i.i.d.) Rician fading channel. The matrices $\mathbf{R}_{\mathrm{TX}} \in \mathbb{C}^{M \times M}$ and $\mathbf{R}_{\mathrm{RX}} \in \mathbb{C}^{N \times N}$ denote the spatial correlation matrices at the TX-SIM and RX-SIM, respectively. The explicit forms of these correlation matrices are provided in~\cite{9716880} as:
%where $\mathbf{G} = \mathbf{R}_{\mathrm{RX}}^{1/2} \bar{\mathbf{G}} \mathbf{R}_{\mathrm{TX}}^{1/2} \in \mathbb{C}^{N \times M}$ is provided in \cite{9714406},  $\bar{\mathbf{G}}$ denotes an independent and identically distributed (i.i.d.) Rician channel, $\mathbf{R_{TX}} \in \mathbb{C}^{M \times M} $ is  the spatial correlation matrix at the TX-SIM, $\mathbf{R_{RX}} \in \mathbb{C}^{N \times N} $ is the spatial correlation matrix at the RX-SIM, which are provided in \cite{9716880} as
\begin{equation}
\begin{aligned}
\left[\mathbf{R}_{\mathrm{TX}}\right]_{\mathrm{m_A},\mathrm{m_A}^{\prime}}
&= \operatorname{sinc}\left( \frac{2 \mathrm{r}_{\mathrm{m_A}, \mathrm{m_A}^{\prime}}^l}{\lambda} \right),
\quad \mathrm{m}, \mathrm{m}^{\prime} \in [1, M], \\
\left[\mathbf{R}_{\mathrm{RX}}\right]_{\mathrm{n_B},\mathrm{n_B}^{\prime}}
&= \operatorname{sinc}\left( \frac{2 \mathrm{r}_{\mathrm{n_B}, \mathrm{n_B}^{\prime}}^k}{\lambda} \right),
\quad \mathrm{n}, \mathrm{n}^{\prime} \in [1, N].
\end{aligned}
\end{equation}

The channel matrix describing the electromagnetic propagation from the output metasurface layer of the TX-SIM to the input metasurface layer of the RX-SIM can be expressed as
\begin{equation}
\mathbf{\bar G}^{\mathrm{H}}=\sqrt{\gamma}\left(\sqrt{\frac{\kappa}{1+\kappa}} \mathbf{h}_{\mathrm{LoS}}^{\mathrm{H}}+\sqrt{\frac{1}{1+\kappa}} \mathbf{h}_{\mathrm{NLoS}}^{\mathrm{H}}\right) \in \mathbb{C}^{N \times M},
\end{equation}
where $\kappa$ is the Rician factor, $\gamma=10^{-\frac{PL(d_{eVTOL})}{10}}$ is the channel gain, $d_{eVTOL}$ is the distance between eVTOL$_A$ and eVTOL$_B$. $\mathbf{h_{LoS}^H} \in \mathbb{C}^{N \times M} $ is the Line-of-Sight (LoS)  coefficient and $\mathbf{h_{NLoS}^H} \in \mathbb{C}^{N \times M} $ is the Non LoS (NLoS) coefficient. ${\rm PL(d_{eVTOL})}$ can be provided by
\begin{align}
&{\rm PL}(d_{\rm eVTOL_A}) = \text{PL}(d_0) + 10\eta \log_{10}\left(\frac{d_{\text{eVTOL}}}{d_0}\right),\nonumber\\
&{\rm PL}(d_0) = 20 \log_{10}\left(\frac{4\pi d_0}{\lambda}\right), \ \
\lambda = \frac{c}{f},
\end{align}
where $\eta$ denotes the path loss exponent, $\text{PL}(d_0)$ represents the reference path loss at a distance $d_0$, and $d_0$ is the reference distance. The parameter $\lambda$ is the carrier wavelength, $f$ is the operating frequency, and $c$ is the speed of light in free space, with $c = 3 \times 10^8$ m/s.

The received channel vector is modeled as
\begin{equation}
    \mathbf{y} = \mathbf{H} \mathbf{x} + \mathbf{n},
\end{equation}
where $\mathbf{x} \in \mathbb{C}^{S \times 1}$ denotes the transmitted signal vector, and $\mathbf{n} \in \mathbb{C}^{S \times 1}$ is the additive noise vector, modeled as a circularly symmetric complex Gaussian random vector, i.e., $\mathbf{n} \sim \mathcal{CN}(0, N_0 \mathbf{I})$, where $N_0$ denotes the noise power spectral density, and the matrix $\mathbf{H} \in \mathbb{C}^{S \times S}$ represents the channel matrix.

Suppose the average transmission power is constrained by $\mathbb{E} \{\mathbf{x}^H \mathbf{x} \} \leq P$, where $P$ denotes the total transmission power, the achievable transmission data can be derived in the following theorem.
\begin{theorem}[Achievable Transmission Data Rate]
Given the HMIMO channel matrix $\mathbf{H} \in \mathbb{C}^{S \times S}$ and the identity matrix $\mathbf{I} \in \mathbb{C}^{S \times S}$, the achievable  transmission data rate $v_{\text{data}}$ as given in~\cite{10534211} can be expressed as
\begin{equation}
\begin{aligned}
v_{\text{data}} = \log_2 \det\left( \mathbf{I} + \frac{\mathbf{H} \mathbf{Z} \mathbf{H}^H}{N_0 B} \right),
\end{aligned}
\end{equation}
where $ \mathbf{Z} = \mathbb{E} \{\mathbf{x} \mathbf{x}^H \} \in \mathbb{C}^{S \times S}$ is a positive semi-definite matrix denotes the covariance matrix of the transmitted signal $\mathbf{x}$.
\end{theorem}

% Suppose the average transmission power is constrained by $\mathbb{E} \{\mathbf{x}^H \mathbf{x} \} \leq P$, where $P$ denotes the total transmission power. The covariance matrix of the transmitted signal $\mathbf{x}$ is defined as $ \mathbf{Z} = \mathbb{E} \{\mathbf{x} \mathbf{x}^H \} \in \mathbb{C}^{S \times S}$ is a positive semi-definite matrix. 
% Given the HMIMO channel matrix $\mathbf{H} \in \mathbb{C}^{S \times S}$ and the identity matrix $\mathbf{I} \in \mathbb{C}^{S \times S}$, the achievable  transmission data rate $v_{\text{data}}$ as given in \cite{10534211} can be expressed as
% \begin{equation}
% \begin{aligned}
% v_{\text{data}} = \log_2 \det\left( \mathbf{I} + \frac{\mathbf{H} \mathbf{Z} \mathbf{H}^H}{N_0 B} \right).
% \end{aligned}
% \end{equation}

\section{ Problem Formulation and Optimization}
\label{sec:Problem Formulation and Optimization}
\subsection{\textit{Problem Formulation}}
This section aims to maximize the data transmission rate $v_{\text{data}}$ in the HMIMO channel while minimizing the propagation delay $t_d$, thereby achieving the joint objective of reducing both the probabilistic delay upper bound $\bar{D} = P\{D > T\}$ and the total end-to-end delay, $T$. Given that the transmission delay, $D_2$ is fixed, and assuming an expected queuing delay denoted by $t_b$, the total end-to-end delay can be expressed as:
\begin{equation}
\begin{aligned}
\begin{gathered}
\mathcal{P}1:\min_{v_{\text{data}},\, t_d} \quad 
f(v_{\text{data}}, t_d) = \mathrm{e^{-\frac{v_{\text{data}} B - \delta_d l_d}{l_d} t_b}+\mathrm{e^{-\frac{v_{\text{data}} B }{S l_d}t_d}}}
+  t_d\\
\text { s.t. } \ \ \  \mathrm{t}_{\mathrm{d}}=\mathrm{T}-\mathrm{D}_2-\mathrm{t}_{\mathrm{b}}, \\
\theta_{\mathrm{m}}^1 \in(0,2 \pi],  \ \xi_{\mathrm{n}}^{\mathrm{k}} \in(0,2 \pi].
\end{gathered}
\end{aligned}
\end{equation}
%where $v_{\text{data}} = \log_2 \det\left( \mathbf{I} + \frac{\mathbf{H} \mathbf{Z} \mathbf{H}^H}{N_0 B} \right)$, $\mathbf{H}=\mathbf{Y G X}=\mathbf{U}^{\mathrm{l}} \boldsymbol{\Psi}^1 \ldots \mathbf{U}^{\mathrm{K}-1} \boldsymbol{\Psi}^{\mathrm{K}-1} \mathbf{U}^{\mathrm{K}} \boldsymbol{\Psi}^{\mathrm{K}} \mathbf{G} \boldsymbol{\Phi}^{\mathrm{L}} \mathbf{W}^{\mathrm{L}} \ldots \boldsymbol{\Phi}^2 \mathbf{W}^2 \boldsymbol{\Phi}^1 \mathbf{W^1} \in \mathbb{C}^{\mathrm{S} \times \mathrm{S} }$ is the HMIMO channel matrix, $\mathbf{Z}$ is the transmit signal covariance matrix, $N_0$ is the power spectral density of the noise, and $B$ is the bandwidth, $\mathbf{I} \in \mathbb{C}^{S \times S}$ denotes the identity matrix. 
Note that, to simplify the optimization problem, we assume that the transmit power is uniformly allocated among all data streams.

To address the complexity of this non-convex optimization problem, we adopt a hybrid solution strategy. First, the BCD method is employed to decompose the problem, thereby reducing its computational difficulty by alternately optimizing over different subsets of variables. Next, the SDR technique is applied to convert the resulting non-convex subproblems into convex formulations, which can be efficiently solved via Semidefinite Programming (SDP). To recover a feasible solution from the relaxed SDR output, the Gaussian Randomization method is subsequently used for post-processing. This integrated framework—comprising BCD-based iterative optimization, SDR-based convex relaxation, and Gaussian randomization for solution recovery—strikes a balance between computational efficiency and solution quality, and provides a practical approach for solving the original non-convex problem.

%To address this complex non-convex optimization problem, we first adopt the BCD method to decompose the problem, reducing the solution difficulty by alternately optimizing different variables. Subsequently, we apply the SDR technique to transform the non-convex constrained problem into a convex optimization problem, which can then be efficiently solved using Semidefinite Programming (SDP). After obtaining the approximate SDR solution, we further employ the Gaussian Randomization method for post-processing, in order to recover a feasible solution from the SDR output and optimize the final configuration of variables. This strategy integrates BCD-based iterative optimization, SDR relaxation, and Gaussian randomization method-based post-processing, ensuring a good balance between computational feasibility and optimization performance, and provides an efficient solution framework for the complex non-convex problem.

\subsection{\textit{The Proposed BCD Algorithm}}
%Structured variable decomposition is a key strategy for handling high-dimensional non-convex problems in wireless communication. The BCD method, widely used in signal processing and resource allocation, iteratively optimizes one variable block at a time while fixing others, significantly reducing complexity and improving convergence. For example, in RIS-assisted systems, BCD method alternately optimizes transmission power and phase shifts~\cite{8811733}, and in the optimization of MIMO channel capacity, it updates signal covariance and spatial mapping via Singular Value Decomposition (SVD)~\cite{8538875}. However, as a greedy approach, the BCD method may converge to local optima, which makes it often combined with methods like SDR~\cite{5946716} to improve solution quality.

In this work, we introduce the BCD algorithm into the optimization framework. The algorithmic process is illustrated as follows.
\begin{algorithm} 
    \caption{BCD Algorithm} \label{Ag1}
    \textbf{Input}: Propagation delay $t_d$,  transmission rate in HMIMO channel $v_{\text{data}}$
    
    \textbf{Output}: Updated data transmission rate $v_{\text{data}}$, updated propagation delay $t_d$
    
    Initialize the propagation delay $t_d \gets 0.6s$

    \For{$iteration=1,2,\cdots,max\_iterations$}{
        Fix $t_d$, calculate the rate in the HMIMO channel, and update the phase shifts of each metasurface layer on SIM by using SDR method;\\

        Fix $v_{\text{data}}$, optimize $t_d$ using an optimizer and update the total end-to-end delay $T$;\\

        Calculate the result according to the objective function of $f_{v_{data},t_d}$;  }
        
   \Return{the value of regret, updated propagation delay $t_d$, upper probabilistic delay bound}
\end{algorithm}\\

% First, by fixing $t_d$, the original problem $\mathcal{P}1$ is transformed into an optimization problem $\mathcal{P}2$ solely with respect to $v_{\text{data}}$. Although this problem remains complex and non-convex, it will be addressed in the next subsection~\ref{subsec:The Proposed Semidefinite Relaxation Method} , using the SDR method. 
As the first step, by fixing the propagation delay $t_d$, the original problem $\mathcal{P}1$ is transformed into a simplified subproblem $\mathcal{P}2$, which involves optimization solely with respect to the data transmission rate $v_{\text{data}}$. Although $\mathcal{P}2$ remains a non-convex and computationally challenging problem, it will be addressed in the subsequent subsection~\ref{subsec:The Proposed Semidefinite Relaxation Method} using the Semidefinite Relaxation technique.
\begin{equation}
\begin{aligned}
\begin{gathered}
\mathcal{P} 2: \min _{\mathrm{v}_{\text {data}}}\quad f\left(\mathrm{v}_{\text {data}}, \mathrm{t}_{\mathrm{d}}\right)=\mathrm{e^{-\frac{v_{\text{data}} B - \delta_d l_d}{l_d} t_b}}+\mathrm{e^{-\frac{v_{\text{data}} B }{S l_d}t_d}}+\mathrm{C}_1 \\
\quad \text { s.t. } 
\theta_{\mathrm{m}}^1 \in(0,2 \pi], \quad\xi_{\mathrm{n}}^{\mathrm{k}} \in(0,2 \pi],
\end{gathered}
\end{aligned}
\end{equation}
%where we define $C_1 = \rho t_d$. Assuming that $v_{\text{data}}$ has been updated in the current iteration, we further fix its value to optimize $t_d$, leading to the reformulated subproblem $\mathcal{P}3$. 
where $C_1 = \rho t_d$ be a predefined constant related to the propagation delay. Assuming that $v_{\text{data}}$ has been updated in the current iteration, we subsequently fix its value to optimize $t_d$, resulting in the reformulated subproblem denoted as $\mathcal{P}3$.

\begin{equation}
\begin{aligned}
\begin{gathered}
\mathcal{P} 3: \min _{\mathrm{t}_{\mathrm{d}}} \quad f \left(\mathrm{t}_{\mathrm{d}}\right)=\mathrm{e}^{-\frac{\mathrm{v}_{\mathrm{data}} \mathrm{~B}}{\mathrm{~S} \mathrm{l}_{\mathrm{d}}} \mathrm{t}_{\mathrm{d}}}+\rho \mathrm{t}_{\mathrm{d}}+\mathrm{C}_2 \\
\text { s.t. } 
\mathrm{t}_{\mathrm{d}}=\mathrm{T}-\mathrm{D}_2-\mathrm{t}_{\mathrm{b}}.
\end{gathered}
\end{aligned}
\end{equation}

In this case, let $C_2 = e^{-\frac{(v_{\text{data}} B - \delta_d l_d)}{l_d} t_b}$. With this substitution, the subproblem $\mathcal{P}3$ reduces to a basic convex optimization problem. This illustrates the effectiveness of the BCD algorithm in handling this class of structured non-convex problems by decomposing them into tractable subproblems.
\begin{equation}
\begin{aligned}
\nabla f\left(t_d\right) & =-\frac{v_{\text {data }} B}{S l_d} e^{-\frac{v_{d a t a} B}{S l_d} t_d}+\rho ,\\
t_d & =-\frac{S l_d}{v_{\text {data }} B} \ln \left(\frac{\rho S l_d}{v_{\text {data }} B}\right).
\end{aligned}
\end{equation}

We directly derive the closed-form solution for the optimal propagation delay $t_d$. Once the optimal $t_d$ is obtained, the optimal total delay $T$ can be computed by the constraint $T = D_2 + t_b + t_d$, and the probabilistic delay bound $\bar{D}$ can be subsequently evaluated.

\subsection{\textit{The Proposed Semidefinite Relaxation Method}}
\label{subsec:The Proposed Semidefinite Relaxation Method}

%The SDR method is a convex optimization technique widely used for solving non-convex optimization problems. It has found extensive applications in wireless communication, signal processing, and resource allocation. In wireless communication optimization problems, complex non-convex constraints such as rank constraints, integer constraints, or non-linear constraints often arise. Solving the original problem directly usually involves high computational complexity and makes it difficult to obtain the global optimum. The SDR method alleviates this challenge by relaxing the original non-convex problem into an SDP problem, which can be efficiently solved using existing convex optimization tools such as CVX, MOSEK, or SDPT3. This approach yields approximate optimal solutions and, in some special cases, even achieves the global optimum.

%In the domain of wireless communication, the SDR method has been widely applied in various scenarios such as MIMO channel estimation~\cite{4518126}, channel gain maximization in RIS-aided systems~\cite{8811733}, and MIMO detection~\cite{5946716}.

After applying the BCD method to decompose the original problem, it follows that minimizing the objective function $f(v_{\text{data}})$, while treating all other variables as constants, is equivalent to maximizing the data transmission rate $v_{\text{data}}$ in the HMIMO channel. Accordingly, the original optimization problem $\mathcal{P}2$ can be simplified to the following problem, denoted as $\mathcal{P}4$:
\begin{equation}
\begin{aligned}
\begin{gathered}
\mathcal{P}4: \max _{\phi^{\mathrm{l}}, \psi^{\mathrm{k}}} \quad \mathrm{v}_{\text {data }}=\log _2 \operatorname{det}\left(\mathbf{I}+\frac{\mathbf{H Z H}^{\mathrm{H}}}{\mathrm{~N}_0 \mathrm{~B}}\right) \\
\text { s.t. } \ \
\theta_{\mathrm{m}}^1 \in(0,2 \pi], \quad\quad\xi_{\mathrm{n}}^{\mathrm{k}} \in(0,2 \pi].
\end{gathered}
\end{aligned}
\end{equation}
Clearly, to maximize $v_{\text{data}}$, it is essential to maximize the effective channel gain characterized by $\mathbf{H}\mathbf{Z}\mathbf{H}^H$. Given that the channel matrix $\mathbf{H}$ is directly influenced by the phase shifts within the transmission coefficient matrices, the optimization problem can ultimately be reformulated as one targeting the HMIMO channel matrix $\mathbf{H}$. In other words, by judiciously controlling the phase shifts, the transmission rate of the HMIMO channel can be enhanced, which is critical for improving overall system performance. Another definition of $\mathbf{H}$ is provided as follows
%Clearly, if the goal is to maximize $v_{\text{data}}$, it is necessary to maximize the channel gain $\mathbf{HZH}$. Since the channel matrix $\mathbf{H}$ is directly related to the phase shifts of the transmission coefficient matrix, the problem can ultimately be transformed into an optimization problem for the HMIMO channel matrix $\mathbf{H}$. In other words, by effectively controlling the phase shifts, the transmission rate of the HMIMO channel can be improved, which is the key to improving the system performance. Give another definition of $\mathbf{H}$:
\begin{equation}
\begin{aligned}
\mathbf{H} 
% &= \mathbf{Y} \mathbf{G} \mathbf{X} \\
&= \mathbf{U}^1 \boldsymbol{\Psi}^1 \cdots \mathbf{U}^K \boldsymbol{\Psi}^K \, 
   \mathbf{G} \, 
   \boldsymbol{\Phi}^L \mathbf{W}^L \cdots \boldsymbol{\Phi}^1 \mathbf{W}^1 \\
&= \boldsymbol{h}_{\mathrm{SIM}, r}^\mathrm{H} \, \boldsymbol{\Omega} \, \boldsymbol{h}_{t, \mathrm{SIM}},
\end{aligned}
\label{eq:H}
\end{equation}
where $\mathbf{h}_{\text{SIM},r}$ denotes the channel coefficient matrix from the $p$-th metasurface layer of the SIM to the receiver, with $p \in \left[1, \max\left(L, K\right)\right]$. The transmission coefficient matrix is defined as $\mathbf{\Omega} = \text{diag}(e^{j\Omega_1}, \dots, e^{j\Omega_{\text{a}}}) \in \mathbb{C}^{\text{a} \times \text{a}}$, where $\Omega_{\mathrm{a}}$ represents the phase shifts of the single-layer SIM subject to optimization. Depending on whether the phase shifts of the TX-SIM or RX-SIM are being optimized, we set $\Omega = \boldsymbol{\phi}$ or $\Omega = \boldsymbol{\psi}$, respectively. The parameter $\mathrm{a}$ denotes the number of meta-atoms in the metasurface layer, which varies according to whether the transmission coefficient matrix corresponds to the TX-SIM or RX-SIM. Furthermore, $\mathbf{h}_{\mathrm{t,SIM}}$ denotes the channel coefficient matrix from the transmitter to the $p$-th metasurface layer of the SIM. The channel coefficient matrices under different optimization objectives are expressed as following two scenarios,

(i) if the phase shifts of TX-SIM is being optimized, then we have
\begin{equation}
\begin{aligned}
\boldsymbol{h}_{S I M, r} &= \mathbf{Y} \mathbf{G} \boldsymbol{\Phi}^L \mathbf{W}^L \cdots \boldsymbol{\Phi}^{p+1} \mathbf{W}^{p+1}, \\
\boldsymbol{h}_{t, SIM} &= \mathbf{W}^p \boldsymbol{\Phi}^{p-1} \mathbf{W}^{p-1} \cdots \boldsymbol{\Phi}^1 \mathbf{W}^1. 
\label{eq:h_SIM_r}
\end{aligned}
\end{equation}

(ii) if the phase shifts of RX-SIM is being optimized, then we have
\begin{equation}
\begin{aligned}
\boldsymbol{h}_{S I M, r} &= \mathbf{U}^1 \boldsymbol{\Psi}^1 \cdots \mathbf{U}^{p-1} \boldsymbol{\Psi}^{p-1} \mathbf{U}^p ,\\
\boldsymbol{h}_{t, SIM} &= \mathbf{U}^{p+1} \boldsymbol{\Psi}^{p+1} \cdots \mathbf{U}^K \boldsymbol{\Psi}^K \mathbf{G} \mathbf{X}.
\label{eq:h_t_SIM}
\end{aligned}
\end{equation}
It is important to note that the phase shift values contained in the matrices $\mathbf{X}$ and $\mathbf{Y}$ are iteratively updated throughout the progression of the BCD algorithm.

Based on the aforementioned channel model, the original optimization problem $\mathcal{P}4$ can be further reformulated as the optimization problem $\mathcal{P}5$
\begin{equation}
\begin{aligned}
\begin{gathered}
\mathcal{P}5: \max _{\Omega} \sum_{s=1}^S| | \boldsymbol{h}_{S I M, r, s}^H \boldsymbol{\Omega} \boldsymbol{h}_{t, S I M}| |^2 \\
\text { s.t. } 0 \leq \Omega_{\text {a }} \leq 2 \pi \quad\text { a } \in[1, \max \{M, N\}],
\end{gathered}
\label{sdfblaer}
\end{aligned}
\end{equation}
where $\boldsymbol{h}_{\text{SIM}, r, s}^H \in \mathbb{C}^{1 \times \text{a}}$ denotes the channel coefficient vector from the $p$-th metasurface layer of the SIM to the receiver for the $s$-th data flow. The vector $\boldsymbol{v} = [v_1, v_2, \ldots, v_{\text{a}}]^T \in \mathbb{C}^{\text{a} \times 1}$ is defined such that each element $v_{\text{a}} = e^{j \Omega_{\text{a}}}$ corresponds to the phase shift of a meta-atom. The constraint on the phase shifts can be equivalently expressed as a unit-modulus constraint, i.e.,
\(|v_{\text{a}}|^2 = 1\), which leads to the following formulation
\begin{equation}
    \sum_{s=1}^S| | \boldsymbol{h}_{S I M, r, s}^H \boldsymbol{\Omega} \boldsymbol{h}_{t, S I M}| |^2=\sum_{s=1}^S| | \boldsymbol{v}^H \boldsymbol{\Lambda}_s| |^2 ,
\end{equation}
where $\boldsymbol{\Lambda}s = \mathrm{diag}\left( \boldsymbol{h}_{\text{SIM}, r, s}^H \right)$ and $\boldsymbol{h}_{t, \text{SIM}} \in \mathbb{C}^{\text{a} \times S}$, thereby leading to the formulation of problem $\mathcal{P}6$ as shown as:
%where \(\boldsymbol{\Lambda}_s = \text{diag}\left( \boldsymbol{h}_{S I M, r, s}^H \right),\boldsymbol{h}_{t, S I M} \in \mathbb{C}^{\text{atom} \times S}\), thus resulting in the above formulation $\mathcal{P}6$.
\begin{equation}
\begin{aligned}
\begin{gathered}
\mathcal{P}6: \quad \max_{\mathbf{v}} \sum_{s=1}^{S} \mathbf{v}^H \boldsymbol{\Lambda}_s \boldsymbol{\Lambda}_s^H \mathbf{v}\\
\text { s.t. }\left|\mathbf{v}_{\text {a }}\right|^2=1, \text { a } \in[1, \max \{M, N\}].
\label{sdfblaer}
\end{gathered}
\end{aligned}
\end{equation}

This transformation results in a Quadratically Constrained Quadratic Programming (QCQP) problem. By introducing an auxiliary variable $i$, the problem can be equivalently reformulated as
\begin{equation}
\begin{aligned}
\begin{gathered}
\mathcal{P}7: \quad \max_{\mathbf{\bar{v}}} \sum_{s=1}^{S} \mathbf{\bar{v}}^H \mathbf{R}_s \mathbf{\bar{v}} \\
\text { s.t. }\left|\mathbf{v}_{\text {a }}\right|^2=1, \text { a } \in[1, \max \{M, N\}],
\end{gathered}
\end{aligned}
\end{equation}
where 
\[ \mathbf{R}_s = \begin{bmatrix}
\mathbf{\Lambda}_s \mathbf{\Lambda}_s^H & 0 \\
0 & 0
\end{bmatrix}, \quad 
\mathbf{\bar{v}} = \begin{bmatrix}
\mathbf{v} \\
i
\end{bmatrix}.
\]

The problem remains non-convex due to the rank-one constraint. Noting that
\( \mathbf{\bar{v}}^H \mathbf{R}_s \mathbf{\bar{v}} = \operatorname{tr}(\mathbf{R}_s \mathbf{\bar{v}} \mathbf{\bar{v}}^H) \),
we define the matrix $\mathbf{V} = \mathbf{\bar{v}} \mathbf{\bar{v}}^H$, which is positive semidefinite and constrained to have $\mathrm{rank}(\mathbf{V}) = 1$. By applying SDR method, we relax the rank-one constraint, thereby reformulating the problem as:
\begin{equation}
\begin{aligned}
\begin{gathered}
\mathcal{P}8: \quad \max_{\mathbf{V}} \sum_{s=1}^{S} \operatorname{tr}\left(\mathbf{R}_s \mathbf{V}\right) \\
\text{s.t.} \quad \mathbf{V}_{\text{a, a}} = 1, \quad \text{i} \in [1, \max\{\mathbf{M}, \mathbf{N}\} + 1], \ \mathbf{V} \succeq 0.    
\end{gathered}
\end{aligned}
\end{equation}

At this stage, problem $\mathcal{P}8$ can be transformed into a convex SDP problem, which can be efficiently solved using existing convex optimization techniques. In this work, the MOSEK solver is employed to numerically obtain the optimal solution to the SDP problem.

Since the optimal solution obtained via the SDR method is generally a positive semidefinite matrix, it is necessary to extract feasible solutions that satisfy the original problem’s constraints in practical SIM phase shift optimization. To this end, the Gaussian randomization method is commonly employed to recover feasible solutions from the relaxed SDR solution. The core idea of Gaussian randomization involves performing eigenvalue decomposition on the SDR solution, followed by random sampling to approximate a rank-one solution that meets the problem constraints. Owing to its relatively low computational complexity and straightforward implementation, the Gaussian randomization method has been widely adopted in various wireless communication optimization problems. In this work, for problem $\mathcal{P}8$, the specific steps of the Gaussian randomization method are outlined as follows.
%Since the optimal solution obtained by the SDR method is typically a positive semidefinite matrix, and in practical SIM phase shifts optimization, we often need to extract feasible solutions that satisfy the constraints, the Gaussian randomization method is widely used to recover feasible solutions from the SDR solution. The core idea of the Gaussian randomization method is to use eigenvalue decomposition to process the SDR solution and combine it with random sampling methods to find the optimal solution. Due to its low computational complexity and ease of implementation, the Gaussian Randomization method is widely applied in various wireless communication optimization problems. In this paper, for problem $\mathcal{P}8$, the specific steps of the Gaussian randomization method are as follows.

%First, perform eigenvalue decomposition in the SDR solution matrix $\mathbf{V}$, obtaining the eigenvalue vector $\boldsymbol{\sigma}$ and the eigenvector matrix $\mathbf{P} \in \mathbb{C}^{\text{atom} \times \text{atom}}$. Then, generate a normalized complex Gaussian random vector,$\mathbf{r} \in \mathbb{C}^{\text{atom} \times 1}$, and use it to construct the optimized phase shifts vector:
First, perform eigenvalue decomposition on the SDR solution matrix $\mathbf{V}$ to obtain the eigenvalue vector $\boldsymbol{\sigma}$ and the eigenvector matrix $\mathbf{P} \in \mathbb{C}^{\text{a} \times \text{a}}$. Next, generate a normalized complex Gaussian random vector $\mathbf{r} \in \mathbb{C}^{\text{a} \times 1}$, which is then used to construct the candidate optimized phase shift vector as follows:
\begin{equation}
\begin{aligned}
\begin{gathered}
\mathbf{v} = \mathbf{P} \, \text{diag}(\sqrt{\boldsymbol{\sigma}}) \, \mathbf{r}.   
\end{gathered}
\end{aligned}
\end{equation}

This vector is then employed to construct the transmission coefficient matrix
\(\boldsymbol{\Omega} = \text{diag}(\mathbf{v})\) which is subsequently substituted into the expression for the transmission rate in the HMIMO channel as follows:
\begin{equation}
\begin{aligned}
\begin{gathered}
\mathrm{v}_{\text {data }}=\log _2 \operatorname{det}\left(\mathbf{I}+\frac{\mathbf{H Z H}^{\mathrm{H}}}{\mathrm{~N}_0 \mathrm{~B}}\right)  .
\end{gathered}
\end{aligned}
\end{equation}

\section{Performance Evaluation}
\label{sec:Performance Evaluation}
In this section, we validate the effectiveness of the SDR method under various parameter settings, demonstrating its capability to enhance the transmission rate. Furthermore, we analyze the influence of key system parameters on both the regret and transmission rate, and investigate the effects of delay and average packet size on the probabilistic delay bound. Unless otherwise specified, the key parameter settings are summarized in Table~\ref{tab:simulationargs}. Additionally, we compare the performance of the proposed BCD algorithm with that of the low-complexity alternating optimization algorithm in terms of transmission rate.

 \begin{table}[!ht]
    \centering
    \caption{Simulation Parameters}
    \begin{tabular}{|l|l|}
    \hline
        \textbf{Description} & \textbf{Value} \\ \hline
        Number of SIM Metasurface layers & 3 \\ \hline
        Number of Meta-Atoms & 36 \\ \hline
        Frequency \( f \) & 7 GHz \\ \hline
        Thickness of TX-SIM and RX-SIM & 0.1 m \\ \hline
        Wavelength \( \lambda \) & 21.4 mm \\ \hline
        % Initial Transmission Rate in HMIMO Channel & 15 bps/Hz \\ \hline
        Reference Path Loss \( PL_0 \) & 80 dB \\ \hline
        % Path Loss Exponent \( \eta \) & 5 \\ \hline
        Area of Meta-Atom  & 0.01 \( \text{m}^2 \) \\ \hline
        Rician Factor \( \kappa \) & 20 \\ \hline
        % Distance Between \( \text{eVTOL}_A \) and \( \text{eVTOL}_B \) & 800 m \\ \hline
        % Weight Coefficient \( \rho \) & 0.72 \\ \hline
        Expected Waiting Time \( t_b \) & 0.5 s \\ \hline
        Mean of Packet Size \( l_d \) & 100 Mb \\ \hline
        % Number of Packets Sent per Second \( \delta \) & 1 packets/s \\ \hline
        Bandwidth \( B \) & 10 MHz \\ \hline
        Noise Power Spectral Density & -210 dBW/Hz \\ \hline
        Transmission Power \( P \) & 0.01 W \\ \hline
        Expected Propagation Delay \( t_d \) & 0.6 s \\ \hline
        % Number of Iterations & 20 \\ \hline
    \end{tabular}
    \label{tab:simulationargs}
\end{table}

Firstly, we illustrate the impact of the number of data streams on the transmission rate, propagation delay, and regret. The results are presented in Fig.~\ref{Figure_differentdata_vdata} through Fig. ~\ref{Figure_differentdata_regret}. It can be observed that, as the number of data streams increases, the transmission rate improves, the propagation delay decreases, and the regret exhibits a downward trend.
\begin{figure}[h]
\centering
     \includegraphics[width=0.3\textwidth ]{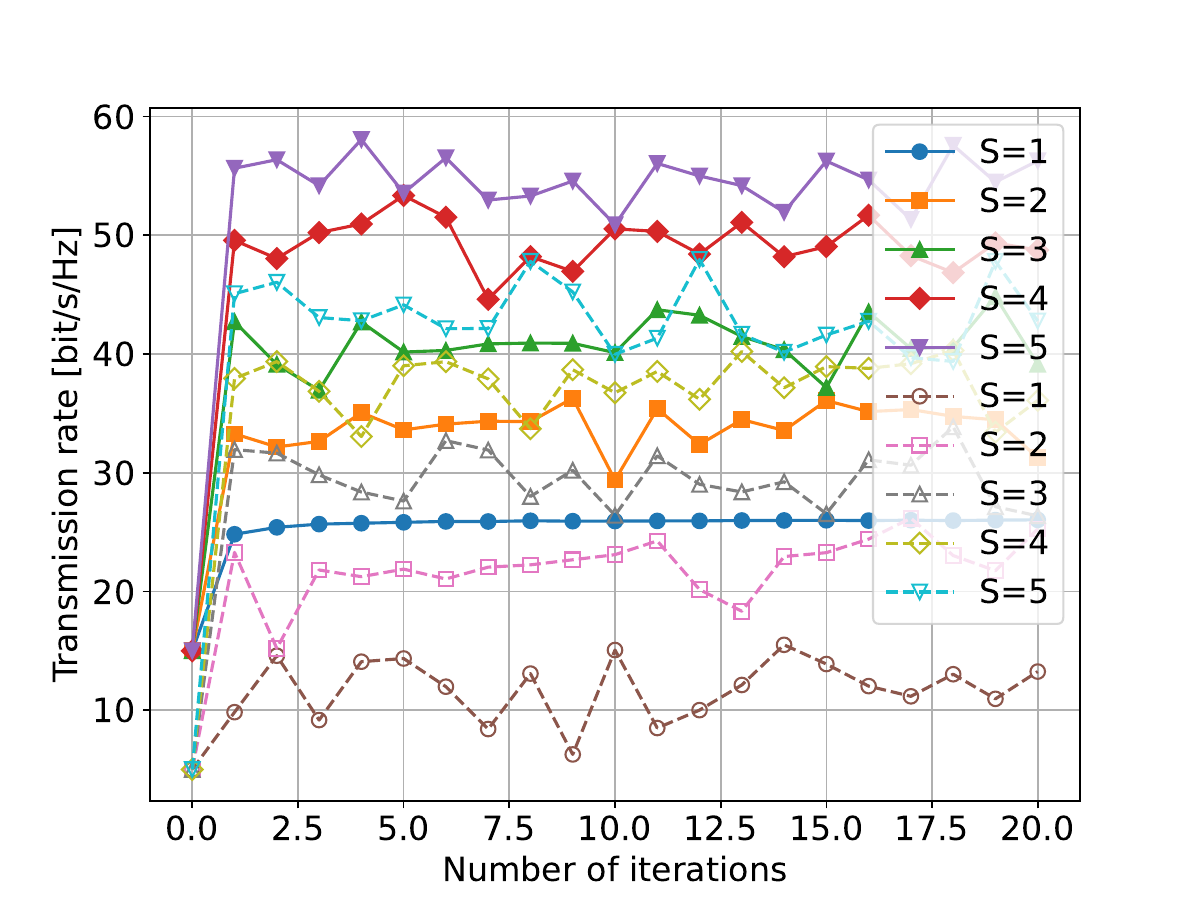} 
     \caption{The relationship between the number of iterations and the data transmission rate $v_{\text{data}}$ under varying numbers of data streams. The solid curves represent the performance of the proposed BCD algorithm, whereas the dashed curves correspond to the results obtained using the low-complexity AO algorithm.}
\label{Figure_differentdata_vdata}
\end{figure}

% \begin{figure}[h]
% \centering
%      \includegraphics[width=0.3\textwidth ]{image/Figure_differentdata_t_d.pdf} 
%      \caption{Relationship between the number of iterations and transmission delay $t_d$/regret $f(v_{data},t_d)$ under different numbers of data streams.} 
% \label{Figure_differentdata_t_d}
% \end{figure}

% \begin{figure}[h]
% \centering
%      \includegraphics[width=0.3\textwidth ]{image/Figure_differentdata_regret.pdf} 
%      \caption{Relationship between the number of iterations and regret $f(v_{data},t_d)$ under different numbers of data streams.} 
% \label{Figure_differentdata_regret}
% \end{figure}
As illustrated in Fig.~\ref{Figure_differentdata_vdata}, although the transmission rate $v_{\text{data}}$ in the HMIMO channel exhibits fluctuations across iterations, it generally follows an upward trend and eventually converges. This observation confirms the effectiveness of the SDR-based optimization approach. In comparison, the low-complexity AO algorithm achieves reduced computational overhead but demonstrates relatively lower optimization performance. Nevertheless, regardless of the algorithm used, the transmission rate increases with the number of data streams. This behavior is primarily attributed to the improved spatial reuse efficiency enabled by a greater number of data streams, thereby enhancing the overall system throughput. Among the tested configurations, the transmission rate attains its maximum when $S = 5$.

%Fig.~\ref{Figure_differentdata_vdata} shows that although the transmission rate in the HMIMO channel fluctuates during iterations, it generally exhibits an increasing trend and eventually stabilizes. This result verifies the effectiveness of the SDR method. In contrast, the low-complexity AO algorithm, while reducing computational load, has relatively weaker optimization performance. However, regardless of the algorithm used, the transmission rate increases as the number of data streams increases. This is mainly because the increase of data streams enhances spatial reuse efficiency, then improving the overall throughput of the system. Under different configurations of data streams, when \( S = 5 \), the transmission rate reaches its highest value.
   Fig.~\ref{Figure_differentdata_t_d} illustrates the variation in propagation delay, with the initial value set to $0.6\text{s}$. Overall, the propagation delay exhibits a decreasing trend as the number of iterations increases, demonstrating the effectiveness of the proposed BCD algorithm in optimizing delay. It is worth noting that when $S = 2$, the propagation delay is slightly higher than that of $S = 1$ (approximately $0.01\text{s}$ difference), suggesting that when the transmission rate is relatively low, increasing $v_{\text{data}}$ has a limited impact on the regret $f(v_{\text{data}}, t_d)$, leading to similar optimization results for $t_d$. In contrast, when $S = 5$, the propagation delay is reduced by approximately $0.1\text{s}$ compared to the case of $S = 1$. This reduction is particularly significant in low-altitude intelligent network scenarios, where low transmission latency is essential to ensure real-time system performance. These results highlight that a slightly increase in the number of data streams can effectively reduce propagation delay and enhance system responsiveness.
%Fig.~\ref{Figure_differentdata_t_d} reflects the trend of transmission delay, with the initial delay set at $0.6s$. In general, as the number of iterations increases, the transmission delay gradually decreases, indicating the effectiveness of the BCD algorithm. Note that when \( S = 2 \), the transmission delay is slightly higher than when \( S = 1 \) (by about $0.01s$), indicating that increasing the transmission rate has less effect on $f(v_{data},t_d)$ when the transmission rate is low, resulting in similar optimization results for $t_d$. Meanwhile, when \( S = 5 \), the transmission delay decreases by about $0.1s$ compared to \( S = 1 \). This reduction is critical for low-altitude intelligent network scenarios, as low delay is a core requirement to ensure real-time system performance. This result further demonstrates that an appropriate increase in the number of data streams can significantly reduce transmission delay.

\begin{figure}[t!]
\centering
\subfloat[{Propagation delay $t_d$.}]{\label{Figure_differentdata_t_d}{\includegraphics[width=0.48\linewidth]{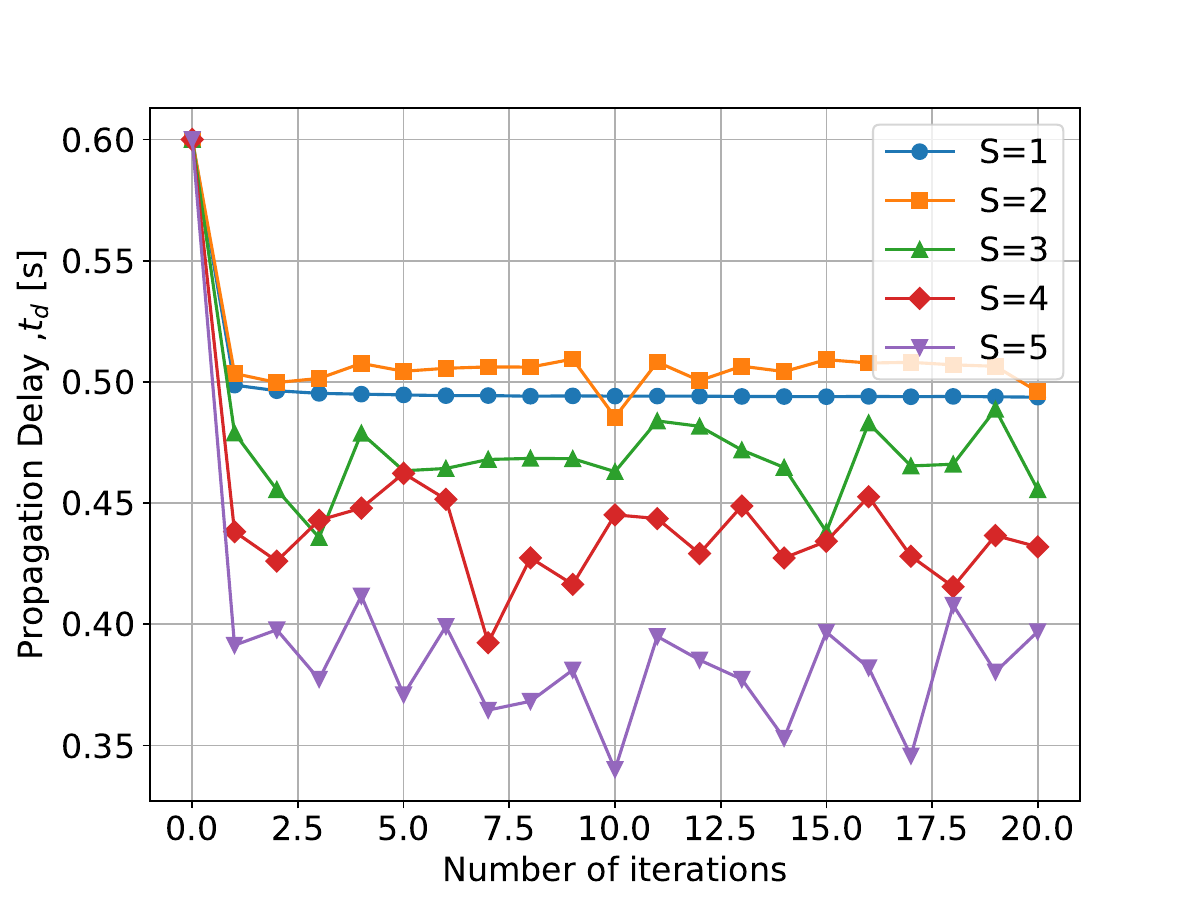}}} 
\subfloat[{Regret $f(v_{data},t_d)$.}]{\label{Figure_differentdata_regret}{\includegraphics[width=0.48\linewidth]{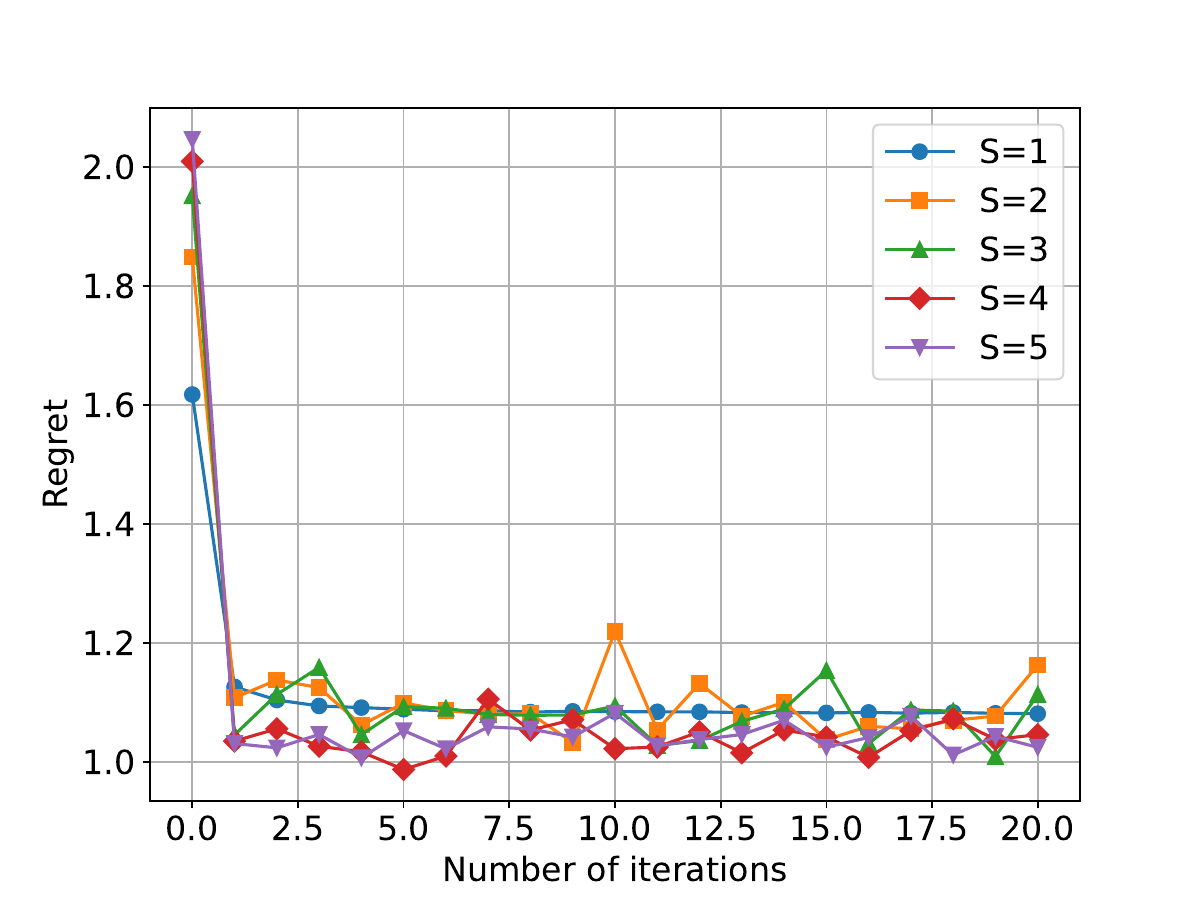}}}
\caption{Relationship between the number of iterations and propagation delay $t_d$/regret $f(v_{data},t_d)$ under different numbers of data streams.}
\label{Figure_3D}
\end{figure}

Fig.~\ref{Figure_differentdata_regret} illustrates the trend of regret as iterations increase. For all values of the number of data streams, the regret consistently decreases with iteration count, though some fluctuations are observed during the convergence process. Notably, for $S = 4$ and $S = 5$, the regret exhibits the most rapid decline, and their final values are nearly indistinguishable. This behavior may be attributed to the structural constraint imposed by the limited number of metasurface layers, as well as the phenomenon of diminishing returns, wherein the incremental performance gain from adding additional data streams becomes marginal beyond a certain threshold. Nevertheless, under identical packet size $l_d$, the configuration with $S = 5$ achieves a lower propagation delay and higher transmission rate compared to other settings. Therefore, $S = 5$ represents a more advantageous configuration for practical deployment in low-altitude intelligent network scenarios.

%Fig.~\ref{Figure_differentdata_regret} shows the trend of regret with respect to the increasing number of iterations. Regardless of the number of data streams, regret decreases with increasing iterations, although there are some fluctuations in the decrease process. When \( S = 4 \) and \( S = 5 \), regret decreases at the fastest rate and the values are very close. This phenomenon may be attributed to the limitation of the number of metasurface layers and the effect of diminishing returns, i.e., when the number of data streams increases to a certain extent, the optimization effect of regret tends to saturate. However, with the same packet size \( l_d \), the transmission delay is lower and the transmission rate is higher when \( S = 5 \), so $S=5$ is a better choice in practical applications.

In Fig.~\ref{Figure_differentatom_vdata} and Fig.~\ref{Figure_differentatom_regret}, we examine the impact of the number of meta-atoms per metasurface layer on the transmission rate and regret, with experimental parameters set as $L = K = 3$. The transmission rate shown is averaged over all iterations.
Fig.~\ref{Figure_differentatom_vdata} shows that the average transmission rate increases as the number of meta-atoms per SIM metasurface layer increases under both BCD algorithm and the low-complexity AO algorithm. Under identical meta-atom settings, the proposed method achieves a 51.47\% average improvement in transmission rate over the low-complexity AO algorithm. In addition, for the same number of meta-atoms, using more data streams leads to a higher average transmission rate, which supports the trend observed in Fig.~\ref{Figure_differentdata_vdata}. However, it is important to note that when $M = N = 9$, the average transmission rate becomes nearly the same for $S \geq 2$. This is because a small number of meta-atoms leads to strong interference, and too many data streams reduce the efficiency of spatial reuse. Therefore, setting the number of meta-atoms to 9 in real-world SIM designs is not recommended, as it may lower system performance.

\begin{figure}[h]
\centering
     \includegraphics[width=0.3\textwidth ]{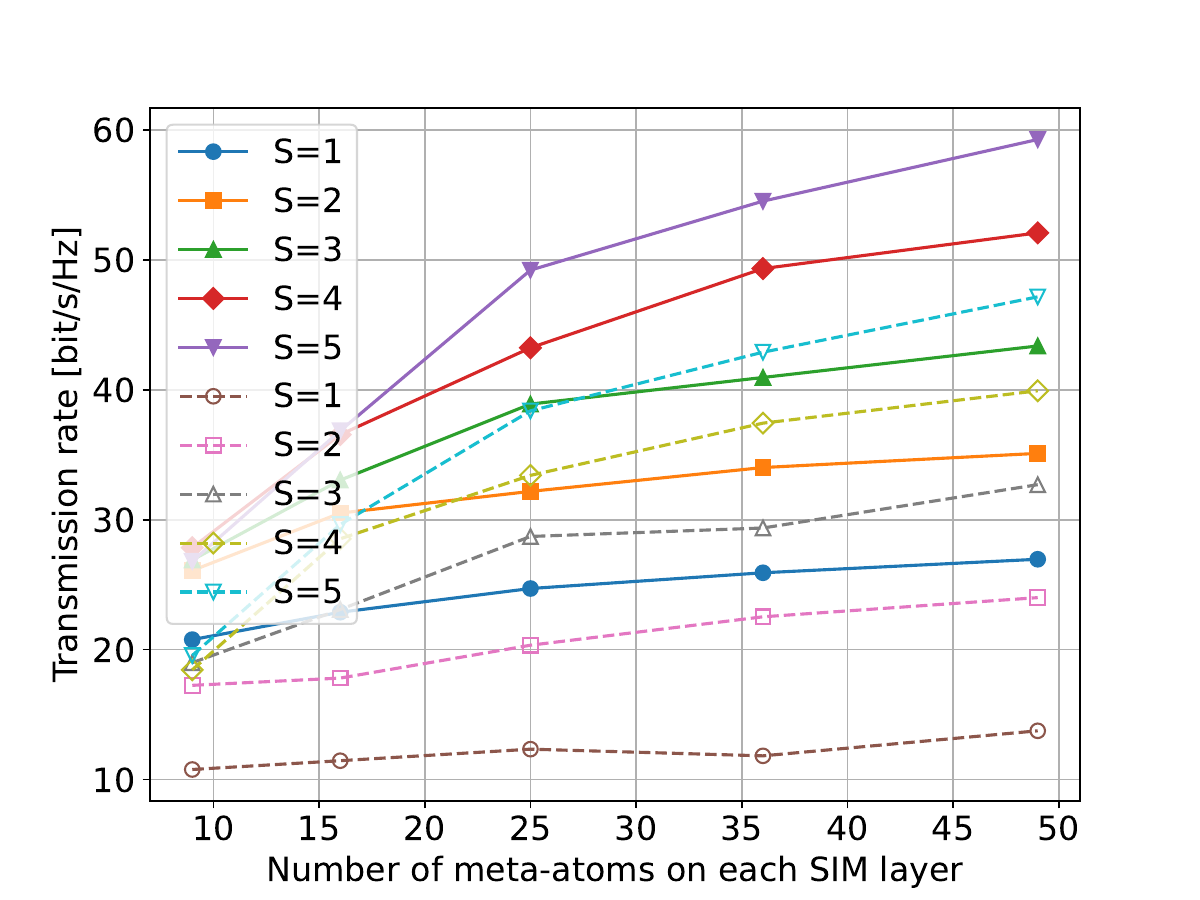} 
     \caption{Relationship between the number of meta-atoms per metasurface layer and $v_{\text{data}}$ under different numbers of data streams. Solid lines represent the results obtained using the BCD algorithm, while dashed lines indicate the performance of the proposed low-complexity AO algorithm.} 
\label{Figure_differentatom_vdata}
\end{figure}

\begin{figure}[h]
\centering
     \includegraphics[width=0.3\textwidth]{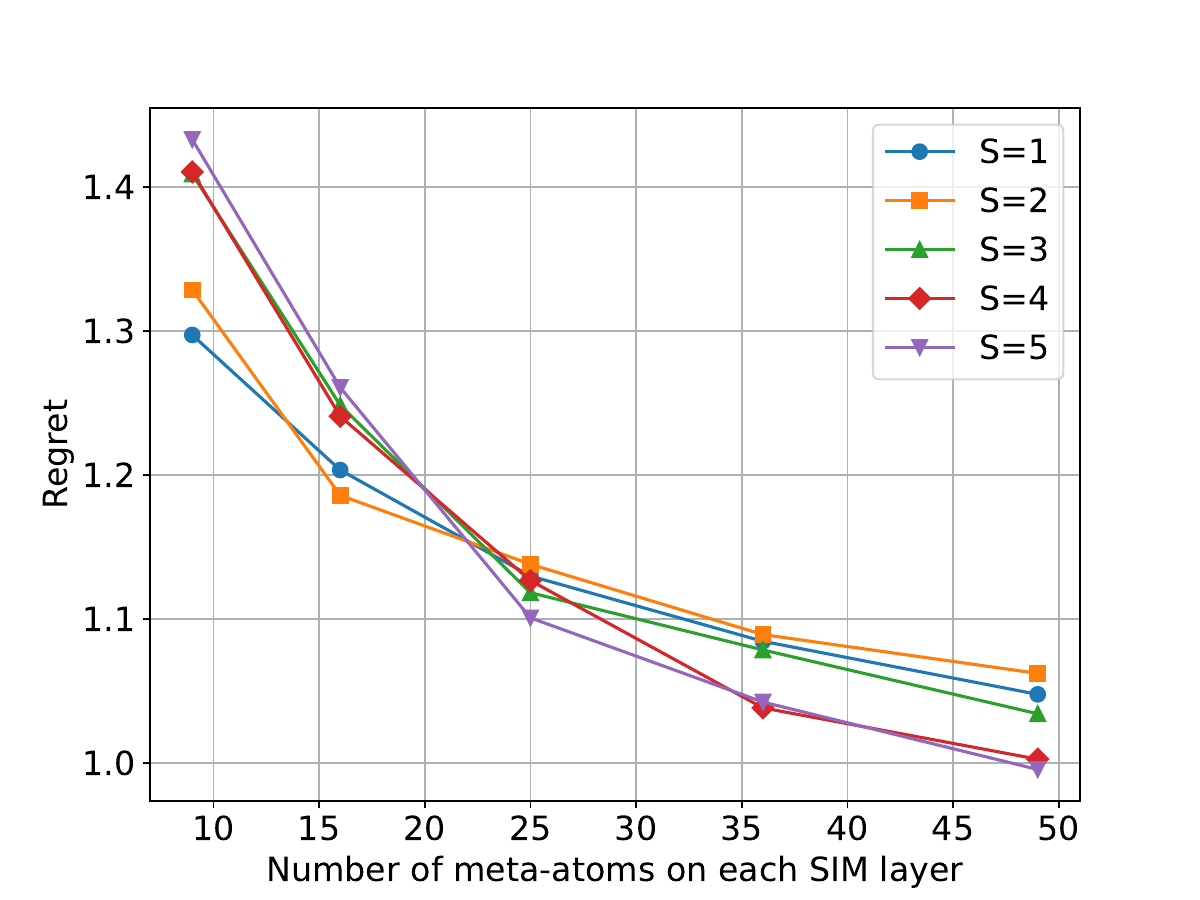} 
     \caption{Relationship between regret and the number of meta-atoms per metasurface layer under different numbers of data streams.} 
\label{Figure_differentatom_regret}
\end{figure}

\begin{figure}[h]
\centering
     \includegraphics[width=0.32\textwidth ]{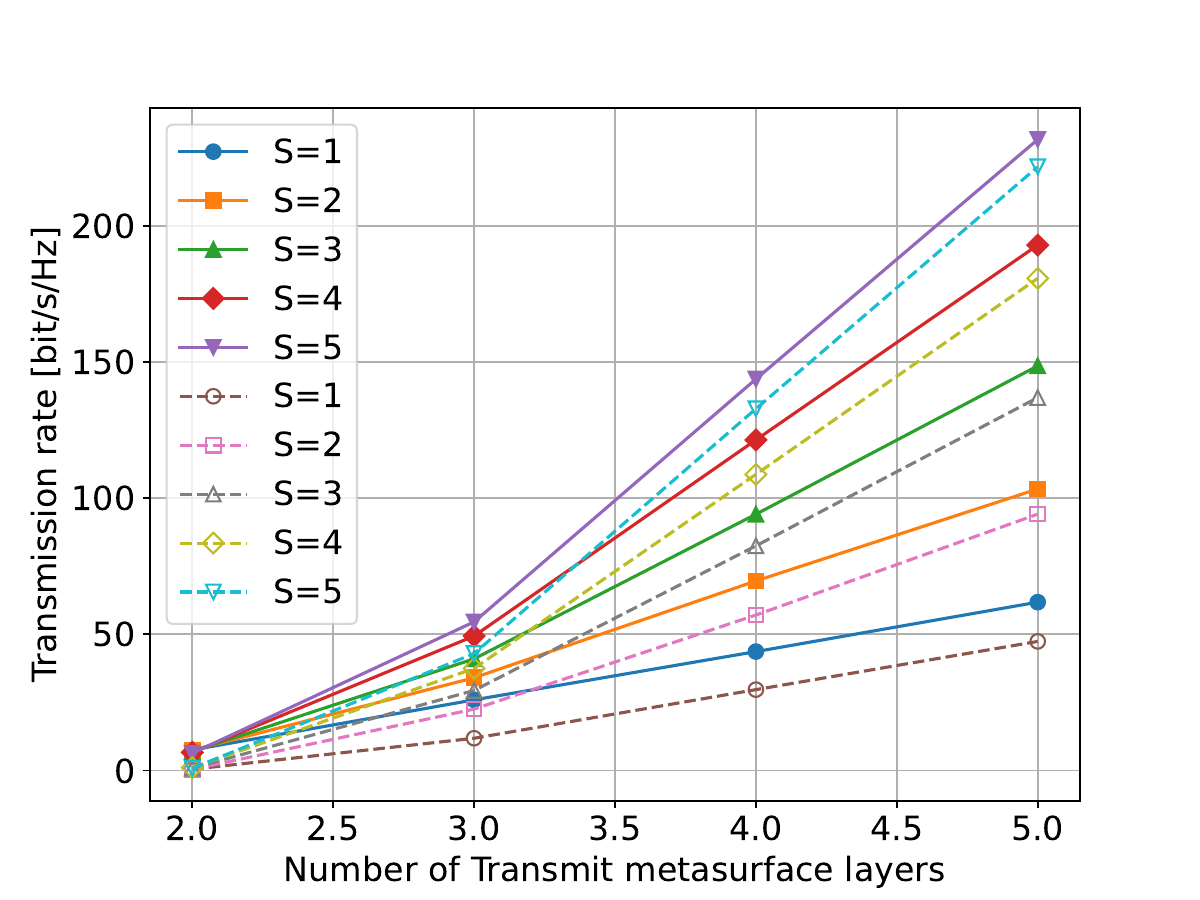} 
     \caption{Relationship between the number of metasurface layers and the transmission rate $v_{\text{data}}$ under different numbers of data streams. Solid lines present the results obtained using the BCD algorithm, while dashed lines show the performance of the low-complexity AO algorithm.} 
\label{Figure_differentlayer_vdata}
\end{figure}

\begin{figure}[h]
\centering
     \includegraphics[width=0.35\textwidth ]{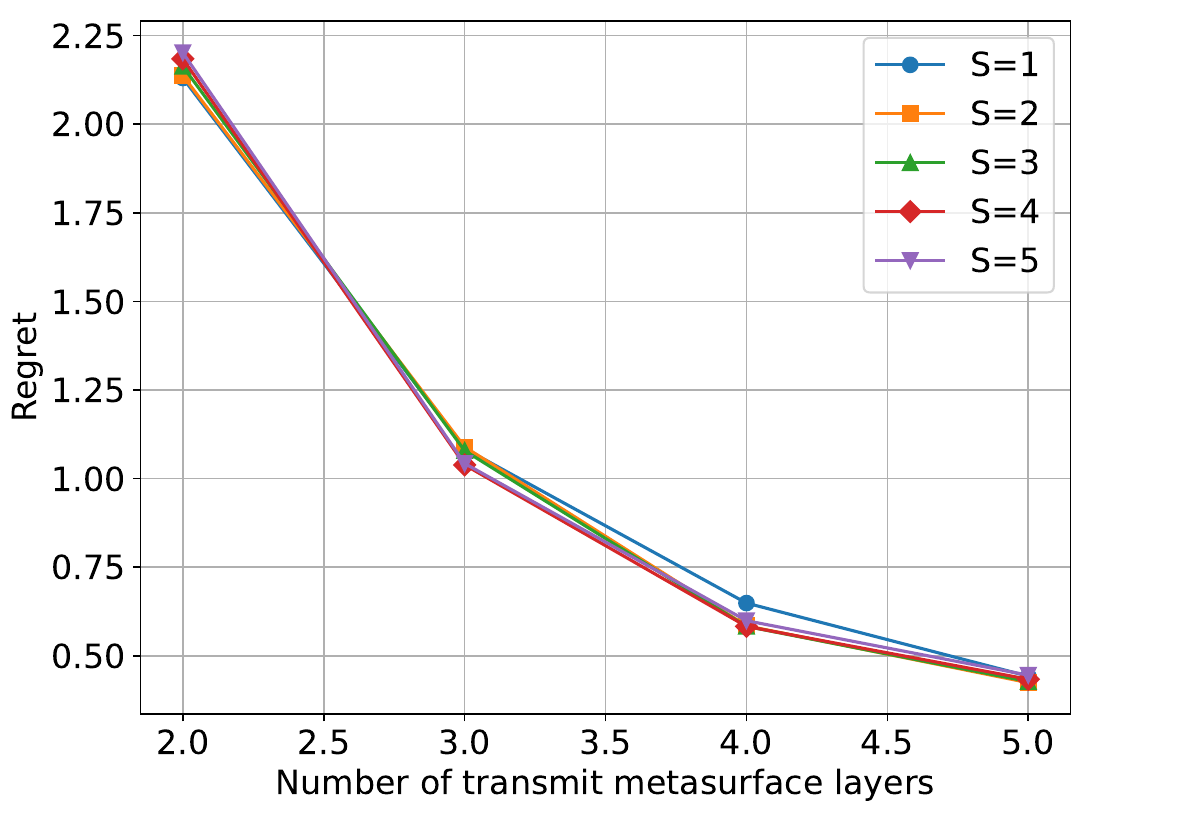} 
     \caption{Relationship between the number of metasurface layers and transmission rate $v_{data}$ under different numbers of data streams.} 
\label{Figure_differentlayer_regret}
\end{figure}
Fig.~\ref{Figure_differentatom_regret} shows that as the number of meta-atoms increases, regret gradually decreases. However, it is important to note that when $M = N < 25$, the regret is lower for $S < 3$ than for $S \geq 3$, whereas when $M = N \geq 25$, higher regret performance is observed for $S \geq 4$. This suggests that when the number of meta-atoms on each metasurface layer is small, increasing the number of data streams does not significantly enhance system performance. The primary reason is that a small number of meta-atoms limits the system’s degrees of freedom, resulting in reduced beamforming capability and less efficient interference management. In this case, fewer data streams lead to lower interference, and the optimization algorithm has limited flexibility, slowing the convergence of regret. In contrast, when the number of meta-atoms is large, the system gains more spatial degrees of freedom, enabling more effective interference mitigation and higher transmission rates. Additionally, the increased flexibility allows the SDR-based optimization to perform better, which accelerates the convergence of regret.
\begin{figure}[t!]
\centering
\subfloat[{S = 1}]{\label{1}{\includegraphics[width=0.33\linewidth]{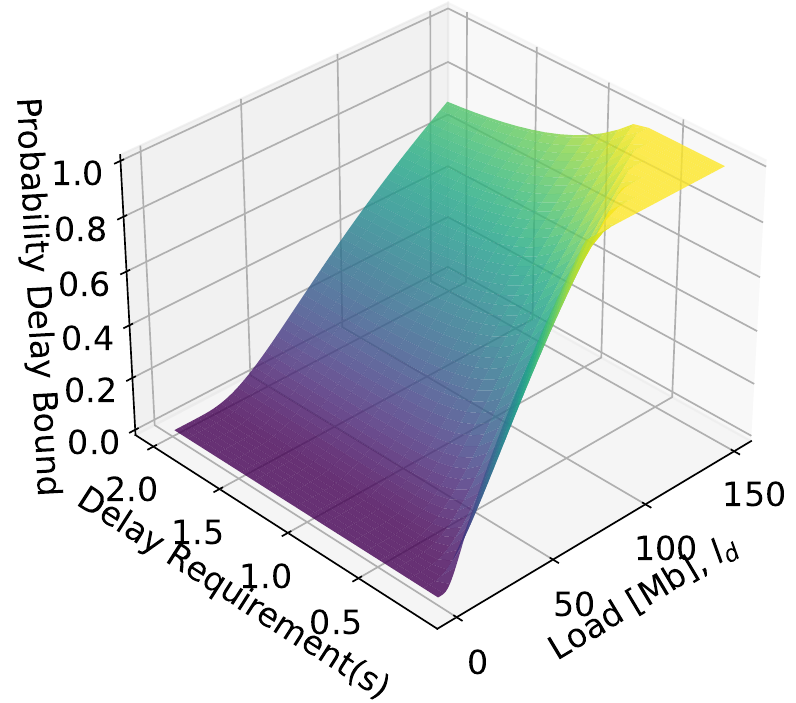}}} 
\subfloat[{S = 3}]{\label{3}{\includegraphics[width=0.33\linewidth]{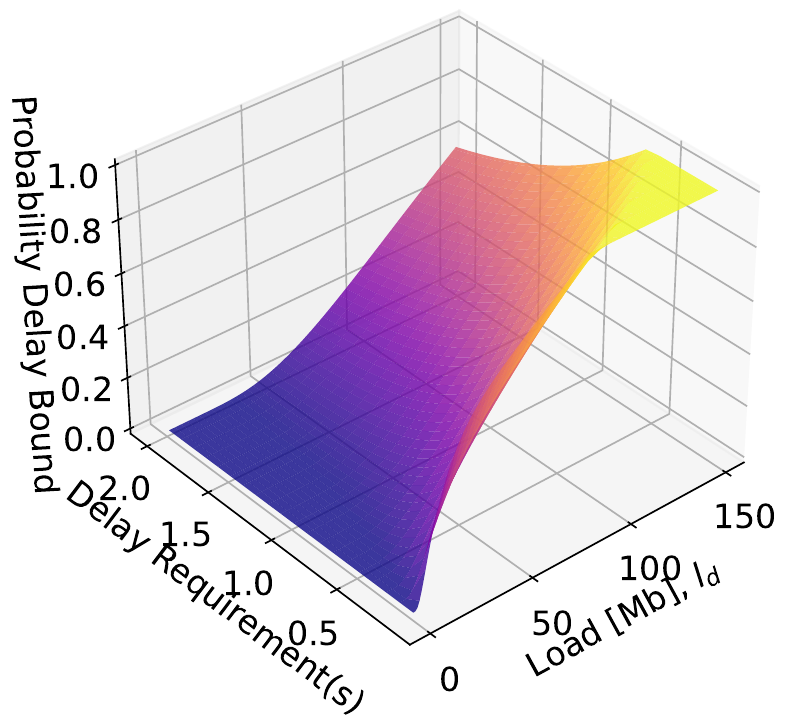}}}
\subfloat[{S = 5}]{\label{5}{\includegraphics[width=0.33\linewidth]{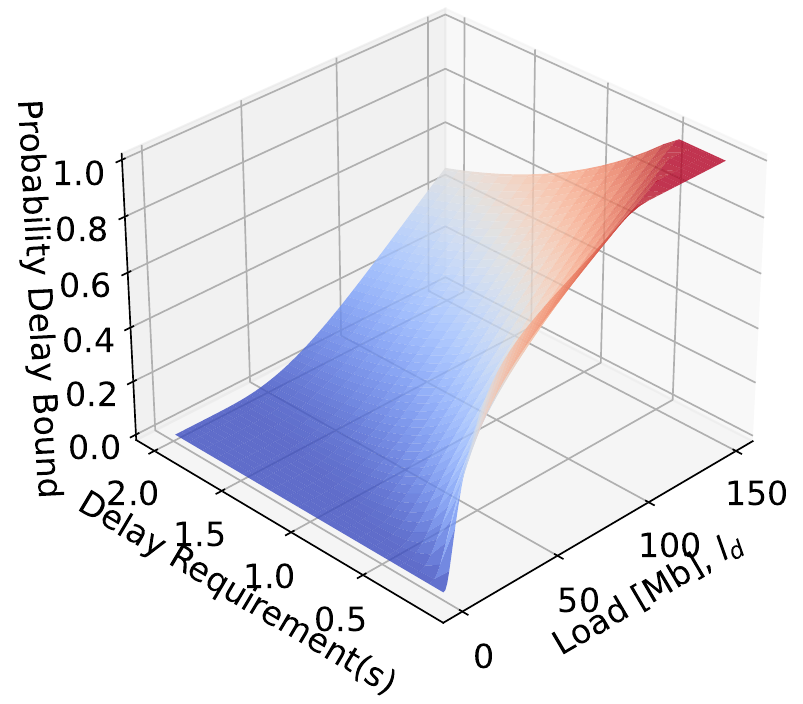}}}
\caption{3D illustration of the probabilistic delay bound with respect to the delay requirement and system load.}
\label{Figure_3D}
\end{figure}
%Fig.~\ref{Figure_differentatom_regret}  shows that as the number of meta-atoms increases, regret gradually decreases. However, it is important to note that when \( M = N < 25 \), the value of regret is less for \( S < 3 \) than for \( S \geq 3 \), while when \( M = N \geq 25 \), \( S \geq 4 \) produces better regret. This indicates that when the number of meta-atoms deployed on the metasurface is low, increasing the number of data streams does not significantly improve the system’s optimization effect. The main reason for this phenomenon is that when the number of meta-atoms is low, the system's degree of freedom is limited and less data streams cause less interference. Additionally, the optimization algorithm has less adjustment space under low degree of freedom, which limits the convergence rate of regret.When the number of meta-atoms is higher, the system's beamforming capability is enhanced, allowing for more effective management of interference between multiple data streams, thereby achieving higher transmission rate. At the same time, SDR optimization can play a greater role with higher degree of freedom, which accelerates the convergence of regret.

Furthermore, Fig.~\ref{Figure_differentlayer_vdata} illustrates that, regardless of whether the SDR method or the low-complexity AO algorithm is employed for optimization, the average transmission rate of the system significantly increases as the number of metasurface layers grows. Notably, this growth rate is substantially higher than that achieved by increasing the number of meta-atoms per metasurface layer. This is primarily because adding more metasurface layers introduces a greater diversity of scattering paths within the channel, thereby enhancing spatial reuse capabilities and effectively boosting the average transmission rate. Consequently, under the same conditions, the highest transmission rate is achieved when $S = 5$, and the SDR-based approach consistently outperforms the low-complexity AO algorithm in terms of performance. However, it is important to note that, given the limited total thickness of the SIM, excessively increasing the number of metasurface layers may lead to greater energy attenuation and increased structural complexity, thereby constraining the potential for further performance improvement.

%Furthermore, Fig.~\ref{Figure_differentlayer_vdata} shows that regardless of whether the SDR method or low-complexity AO algorithm is used for optimization, as the number of metasurface layers increases, the system's average transmission rate improves significantly, and its growth rate is noticeably faster than the improvement brought about by increasing the number of meta-atoms per metasurface layer. This is primarily because adding metasurface layers provides more diverse channel scattering paths,enhancing spatial reuse capability and effectively increasing the system average transmission rate. Therefore, under the same conditions, the transmission rate is highest when \( S = 5 \), and the SDR method consistently outperforms the low-complexity AO algorithm. {\color{red} However, it is important }

% \begin{figure}[h]
% \centering
%      \includegraphics[width=0.3\textwidth ]{image/Figure_probability_t.pdf} 
%      \caption{Relationship between total delay ($t$) and probabilistic delay bound under different load} 
% \label{Figure_probability_t}
% \end{figure}

\begin{figure}[t!]
\centering
\subfloat[{Total delay ($T$).}]{\label{Figure_probability_t}{\includegraphics[width=0.48\linewidth]{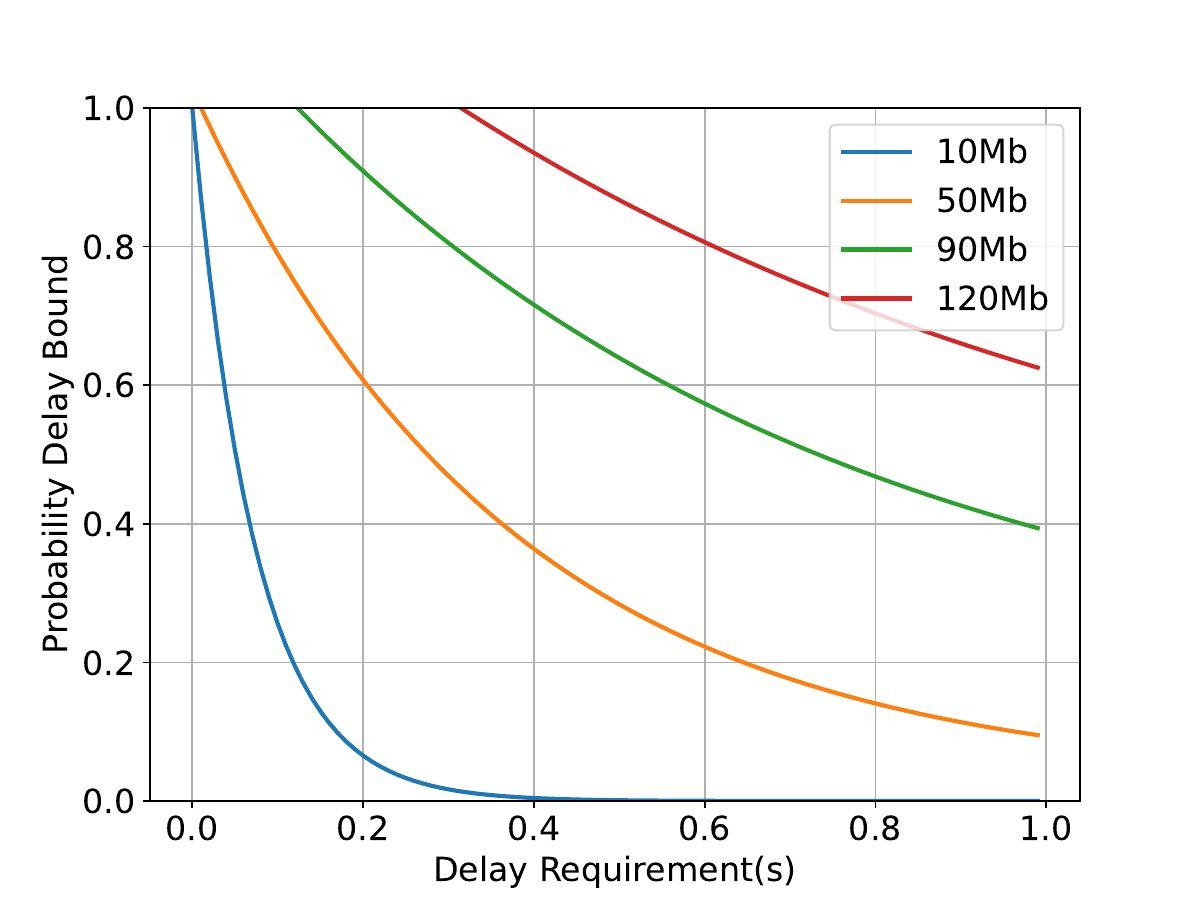}}} 
\subfloat[{Load ($l_d$).}]{\label{Figure_probability_l_d}{\includegraphics[width=0.48\linewidth]{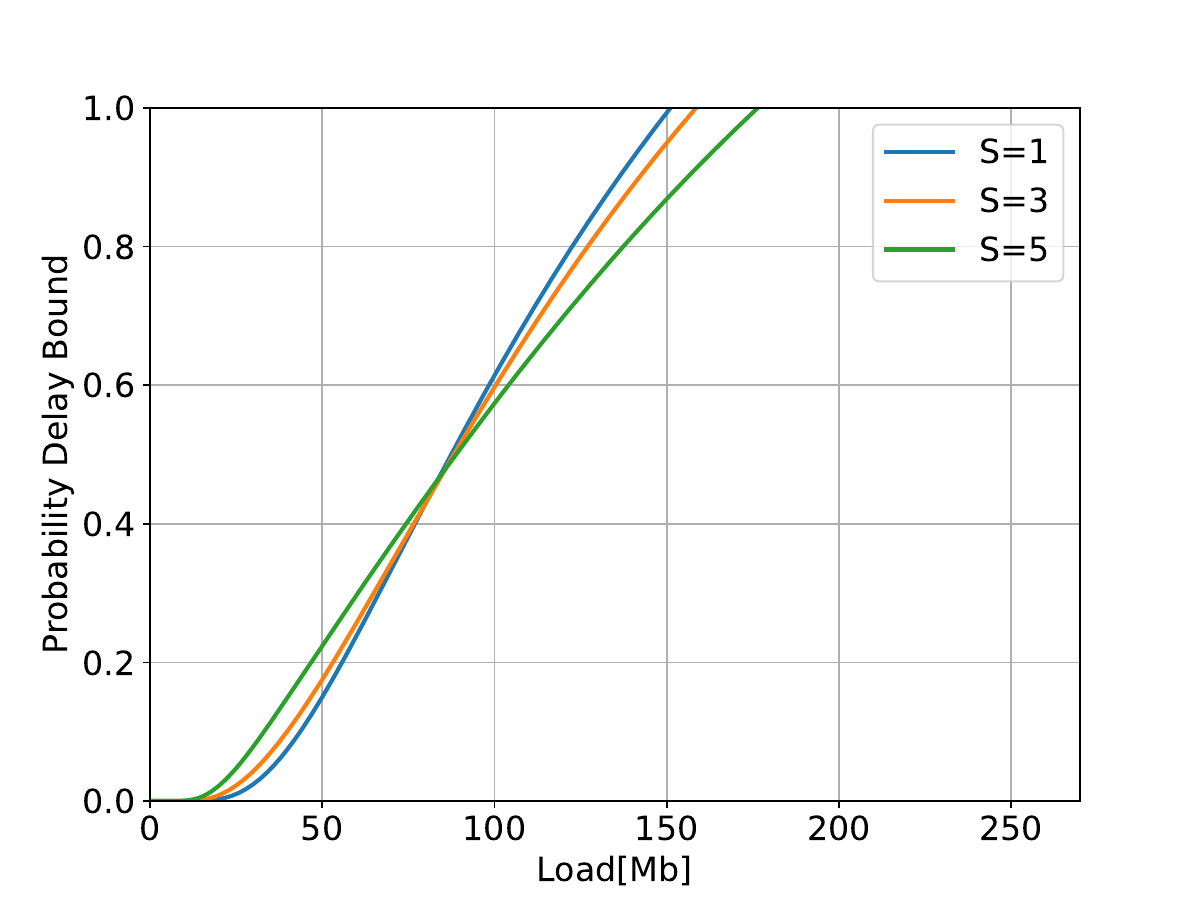}}}
\caption{Relationship between total delay ($T$)/load ($l_d$) and probabilistic delay bound under different load.}
\label{Figure_3D}
\end{figure}

%However, it is important to note that since the total thickness of the SIM is limited, excessively increasing the number of metasurface layers may introduce more energy attenuation or structural complexity, which in turn limits the potential for further performance improvement.

Fig.~\ref{Figure_differentlayer_regret} shows that as the number of metasurface layers increases, the average transmission rate improves while the regret gradually decreases, further confirming the performance enhancement brought by multi-layer structures. However, it is observed that when $L = K = 2$, the regret values across different data stream configurations remain close. As the number of metasurface layers increases to $L = K = 3$ or $4$, the differences in regret become more pronounced across varying data stream quantities, indicating that higher spatial reuse capability leads to improved optimization performance. Nevertheless, when the number of layers increases further to $L = K = 5$, the regret values for different numbers of data streams converge, suggesting that the system is reaching its optimal structural configuration. At this point, the benefits of adding more data streams for regret reduction begin to saturate.

%Fig.~\ref{Figure_differentlayer_regret}  shows that as the number of meta-atom metasurface layers increases, the average transmission rate increases, while regret gradually decreases, further proving the enhancement effect of multi-layer structures on system performance. However, we observe that when \( L = K = 2 \), the differences of regret in all data stream configurations are small. As the number of metasurface layers increases to \( L = K = 3= 4 \), the difference of the number of data streams increases the gap of the value of regret, indicating that higher system spatial reuse capability effectively improves optimization performance of the algorithm. However, when the number of layers further increases to \( L = K = 5 \), the regret values under different data stream quantities almost same, suggesting that with sufficiently high structural degree of freedom, the system is approaching its optimal configuration, and at this point, the further optimization of regret by increasing the number of data streams is nearing saturation.

% \begin{figure}[h]
% \centering
%      \includegraphics[width=0.3\textwidth ]{image/Figure_probability_l_d.pdf} 
%      \caption{Relationship between load ($l_d$) and probabilistic delay bound under different numbers of data streams} 
% \label{Figure_probability_l_d}
% \end{figure}
Fig.~\ref{Figure_3D} provides a comprehensive three-dimensional illustration of the relationship among the probabilistic delay bound, the propagation delay, and the packet size. The parameters are set as follows: under fixed average transmission rates, the average transmission rate is 25.92 bps/Hz for $S=1$, 40.96 bps/Hz for $S=3$, and 54.55 bps/Hz for $S=5$.

%Fig.~\ref{Figure_3D} comprehensively illustrates the relationship between the probabilistic delay bound, the transmission delay, and the packet size in a three-dimensional representation. The parameters are set as follows: Under a fixed average transmission rate, the average transmission rate is 25.92 bps/Hz for \( S=1 \), 40.96 bps/Hz for \( S=3 \), and 54.55 bps/Hz for \( S=5 \). 

From the subplots in Fig.~\ref{Figure_3D}, it can be observed that as the packet size increases, the probabilistic delay bound also increases. This is because larger packets require longer transmission times within the network, which raises the probabilistic delay bound. Additionally, at lower total end-to-end delay $T$, the probabilistic delay bound is higher, since a stricter delay requirement increases the probability of exceeding that threshold.
Moreover, when $S = 5$, the average transmission rate is the highest, resulting in the best performance in terms of both the probabilistic delay bound and propagation delay. This indicates that a higher number of data streams can improve spatial reuse efficiency, thereby reducing end-to-end propagation delay and optimizing the probabilistic delay bound. Conversely, when $S = 1$, the average transmission rate is the lowest, leading to the poorest performance in terms of the probabilistic delay bound and propagation delay.
These results further demonstrate that appropriately increasing the number of data streams can effectively improve the system’s delay performance, enhancing communication reliability and Quality of Service (QoS) assurance.
%From the subgraphs in Fig.~\ref{Figure_3D} , it can be observed that as the packet size increases, the probabilistic delay bound also increases. This is because a larger packet requires longer transmission time in the network, which raises the probabilistic delay bound. Additionally, at lower total end-to-end delay $t$, the probabilistic delay bound is higher. This is due to lower overall delay requirement that leads to higher probability of exceeding that threshold. Moreover, when \( S=5 \), the average transmission rate is the highest, resulting in the best performance in terms of probabilistic delay bound and transmission delay. This indicates that a higher number of data streams can improve spatial reuse rate, thereby reducing end-to-end transmission delay and optimizing the probabilistic delay bound. Conversely, when \( S=1 \), the average transmission rate is the lowest, leading to the worst performance in terms of the probabilistic delay bound and the transmission delay. This further demonstrates that appropriately increasing the number of data streams can effectively improve the system’s delay, enhancing communication reliability and QoS assurance.

Fig.~\ref{Figure_probability_t} illustrates the impact of total delay on the probabilistic delay bound, where the transmission rate is fixed and averaged at 40.86 bps/Hz. As shown in Fig.~\ref{Figure_probability_t}, the probabilistic delay bound $P\{D > T\}$ decreases as the delay threshold $T$ increases. This indicates that the likelihood of packets exceeding the delay threshold diminishes with a larger allowable delay.
Moreover, a larger average packet size $l_d$ leads to a higher probabilistic delay bound, meaning the probability of delay violation increases. Conversely, for a fixed probabilistic delay bound, larger packets require a longer propagation delay. This implies that smaller packet sizes contribute to better delay performance in the HMIMO communication system.
The main reason is that smaller packets allow the HMIMO channel transmission rate $v_{\text{data}}$ to adequately handle incoming traffic, thus reducing propagation delay $t_d$. Consequently, the overall packet delay is mainly influenced by the transmission time within the HMIMO channel rather than by congestion or queuing delays.
These results emphasize that appropriately managing packet size $l_d$ is crucial for optimizing propagation delay $t_d$ and improving the delay performance of the communication system.

%Fig.~\ref{Figure_probability_t} illustrates the impact of total delay on the probabilistic delay bound, where the transmission rate is fixed and averaged, with an average transmission rate of 40.86 bps/Hz. As observed in Fig.~\ref{Figure_probability_t}, the probabilistic delay bound \( P\{D>t\} \) exhibits a decreasing trend as time \( t \) increases. This indicates that as the delay threshold increases, the probability of packet exceeding this delay gradually decreases. Moreover, a larger average packet size \( l_d \) results in a higher probabilistic delay bound, which means the probability of exceeding this delay increases. Conversely, for the same probabilistic delay bound, larger packets require a longer transmission delay. This suggests that smaller packet sizes lead to better delay performance in the HMIMO communication system. The primary reason is that when the packets are smaller, the transmission rate of the HMIMO channel is sufficient to accommodate the incoming traffic, thus reducing the transmission delay in the system. As a result, the delay experienced by packets is mainly determined by their transmission time in the HMIMO channel, rather than by system congestion or queuing delay. These findings highlight that appropriately controlling packet size can help optimize the transmission delay and enhance the delay performance of the communication system.

Fig.~\ref{Figure_probability_l_d} illustrates the relationship between the probabilistic delay bound $P\{D > 0.7\}$ and the average packet size $l_d$, under a fixed average transmission rate. The average transmission rates are set as follows: 25.92 bps/Hz for $S=1$, 40.96 bps/Hz for $S=3$, and 54.55 bps/Hz for $S=5$.
As observed in Fig.~\ref{Figure_probability_l_d}, the probabilistic delay bound increases with the mean packet size $l_d$. This is because larger packets require longer transmission times, resulting in a higher probability of exceeding the delay threshold.
In particular, when the packet size is less than 85 Mb, the probabilistic delay bound for $S=1$ is lower than those for $S=3$ and $S=5$, indicating that $S=1$ achieves the best performance in this range. This occurs because, for small packet transmission, the available channel resources with $S=1$ are sufficient to support the data flow.
However, when the packet size exceeds 85 Mb, the probabilistic delay bound for $S=5$ becomes significantly better than those for other stream numbers, demonstrating the advantage of a larger number of data streams in transmitting large packets.
Specifically, the average packet size at which the probabilistic delay bound reaches 1 (i.e., when $T$ exceeds 0.7 s) is 150 Mb for $S=1$, 158 Mb for $S=3$, and 175 Mb for $S=5$. This indicates that for larger data transmissions, increasing the number of data streams can effectively reduce the end-to-end propagation delay, thereby improving delay performance.

%Fig.~\ref{Figure_probability_l_d} illustrates the relationship between the probabilistic delay bound \( P\{D>0.7\} \) and the average packet size \( l_d \), with a fixed average transmission rate. The average transmission rates are set as follows: 25.92 bps/Hz for \( S=1 \), 40.96 bps/Hz for \( S=3 \), and 54.55 bps/Hz for \( S=5 \).As observed in Fig.~\ref{Figure_probability_l_d} , the probabilistic delay bound increases with the mean size of packet $l_d$. This is because larger packets require longer transmission time, leading to a higher probability of exceeding the delay threshold. In particular, when the packet size is less than 85 Mb, the probabilistic delay bound for \( S=1 \) is lower than that for \( S=3 \) and \( S=5 \), indicating that \( S=1 \) performs best in this range. This occurs because, in the case of small packet transmission and the available channel resources for \( S=1 \) are sufficient to support data transmission. However, when the packet size exceeds 85 Mb, the probabilistic delay bound for \( S=5 \) becomes significantly better than that for other stream numbers, demonstrating the advantage of larger data stream numbers in large packet transmission. Specifically, the average packet size at which the probabilistic delay bound reaches 1 (i.e., when $t$ exceeds $0.7s$) is 150 Mb for \( S=1 \), 158 Mb for \( S=3 \), and 175 Mb for \( S=5 \). This indicates that for larger data transmission, increasing the number of data streams can effectively reduce end-to-end transmission delay, thereby improving delay performance. 

\section{Conclusion}
\label{sec:Conclusion}
This paper proposes a SIM-assisted HMIMO aerial communication system designed to meet the stringent latency and reliability requirements of dynamic AAM scenarios. By integrating the TX-SIM and RX-SIM on eVTOL platforms, the system leverages precoding and combining techniques at both the transmitter and receiver ends to enhance performance. To address the trade-off between propagation delay and the probabilistic delay bound, we formulate a non-convex optimization problem, which is efficiently solved using a BCD framework combined with the SDR method. Further analysis reveals the impact of the number of data streams, meta-atom density, and the number of metasurface layers on delay performance and system regret. Additionally, we investigate the effects of system load and total delay $T$ on the probabilistic delay bound. This work provides practical insights for deploying SIM-assisted HMIMO systems and lays the foundation for future research on scalable multi-eVTOL coordination and learning-based optimization under imperfect CSI.
%In this paper, we proposed a SIM-assisted Holographic MIMO aerial communication system to meet the stringent latency and reliability demands of dynamic AAM scenarios. By integrating TX-SIM and RX-SIM on eVTOLs, the system benefits from the precoding and combination at the transmitter and receiver, respectively. To address the trade-off between transmission delay and probabilistic delay bound, a non-convex optimization problem was formulated and efficiently solved using a BCD framework with SDR method. Simulation results demonstrate that our design significantly improves the transmission rate and reduces the total end-to-end delay. Result analysis further reveals the influence of the number of data streams, meta-atom density, and the number of metasurface layers on delay performance and system regret. Additionally, we explored how system load and total delay \( t \) affect the probabilistic delay bound.This work offers practical insights for deploying SIM-assisted HMIMO systems and paves the way for future research on scalable multi-eVTOL coordination and learning-based optimization under imperfect CSI.

\appendix
\subsection{Proof of Theorem \ref{theorm_barD1}}\label{app_barD1}
Based on SNC, the queuing delay is defined as follows:
% The delay component \( D_1 \) represents the queuing delay. Note that only the queuing delay at the TX-SIM side is considered in this analysis. The queuing delay is defined as:
\begin{equation}
D_1 = \inf\{\tau \geq 0 : A(t) \leq A^*(t + \tau)\}.
\end{equation}
Consequently, the event \( \{D_1 > t_b\} \) implies \( \{A(t) > A^*(t + t_b)\} \), and it follows that:
$\{D_1 > t_b\} \subseteq \{A(t) > A^*(t + t_b)\}$ According to the stochastic arrival curve \( A \sim \langle f, \alpha \rangle \) and the stochastic service curve \(S\sim \langle g, \beta \rangle \), the probabilistic upper bound of the queuing delay can be expressed as~\cite{10039312}: $P\{D_1 > t_b\} $ represents that  the probabilistic delay bound of $D_1$, which is also called an upper bound of the complementary cumulative distribution function (CCDF) of the packets delay $t_b$. i.e. the probability that the data packets wait in the buffer for more than $t_b$.
\begin{equation}
\begin{aligned}
&P\{D_1 > t_b\} \leq P\{A(t) - A^*(t + t_b) > 0\}  \\
% &\leq P\left\{
% \sup_{0 \leq \tau \leq t + t_b} 
% \left[ A(\tau, t) - \alpha(t - \tau) 
% + \alpha(t - \tau) \right. \right.  \\
% &\quad \left. \left.
% - \beta_{\text{data}}(t + t_b - \tau) 
% \right] 
% + A \otimes \beta_{\text{data}}(t + t_b) 
% - A^*(t + t_b) > 0 
% \right\}  \\
&\leq P\left\{
\sup_{0 \leq \tau \leq t} 
\left[ A(\tau, t) - \alpha(t - \tau) \right] 
+ A \otimes \beta_{\text{data}}(t + t_b) 
 \right.  \\
&\quad \left. - A^*(t + t_b)> \inf_{0 \leq \tau \leq t + t_b} 
\left[ \beta_{\text{data}}(t + t_b - \tau) 
- \alpha(t - \tau) \right] 
\right\}  \\
&\leq f \otimes g \left( 
\inf_{0 \leq \tau \leq t + t_b} 
\left[ \beta_{\text{data}}(t + t_b - \tau) 
- \alpha(t - \tau) \right] 
\right)  \\
&\leq f \otimes g \left( 
\beta_{\text{data}}(t_b) 
\right)
\end{aligned}
\label{eq:prob_delay_bound}
\end{equation}

In the last line of \eqref{eq:prob_delay_bound}, we apply a sufficient stability condition: for all \( t \geq 0 \), we have \( \beta(t) \geq \alpha(t) \).

We hold the service curve \( \beta_{\text{data}}(t) \) for \( D_1 \) as
\begin{equation}
\beta_{\text{data}}(t) = v_{\text{data}}  B t ,
\label{eq:service curve}
\end{equation}
where \( v_{\text{data}} \) is the transmission rate in the HMIMO channel,the calculation formula of $v_{data}$ is given in the subsection~\ref{sec:SIM Performance Analysis}. \( B \) denotes the system bandwidth and \( g(x) = 0 \). Here, \( f(x) \) and \( g(x) \) are bounding functions.
Using equations \eqref{eq:prob_delay_bound} and \eqref{eq:service curve}, the probabilistic delay bound of \( D_1 \) becomes
\begin{equation}
P\{D_1 > t_b\} \leq f(v_{\text{data}} B t_b) \triangleq \bar{D}_1(t_b).
\end{equation}
\appendices
\footnotesize
\bibliography{biblio}
\end{document}